\documentclass[11pt]{article}

\usepackage{fullpage}
\usepackage{amsmath,amsfonts,amsthm,mathrsfs,mathpazo,xspace,hyperref,graphicx}
\usepackage{endnotes}
\usepackage{color}
\usepackage{bm}
\usepackage{times}
\usepackage{amssymb,latexsym}
\usepackage{enumitem}
\usepackage{braket}
\usepackage{tikz}
\usepackage{float}
\usepackage{verbatim}

\makeatletter
\newtheorem*{rep@theorem}{\rep@title}
\newcommand{\newreptheorem}[2]{%
\newenvironment{rep#1}[1]{%
 \def\rep@title{#2 \ref{##1}}%
 \begin{rep@theorem}}%
 {\end{rep@theorem}}}
\makeatother

\newtheorem{theorem}{Theorem}[section]
\newreptheorem{theorem}{Theorem}
\newtheorem*{theorem*}{Theorem}

\newtheorem{lemma}[theorem]{Lemma}
\newtheorem{claim}[theorem]{Claim}

\newtheorem{corollary}[theorem]{Corollary}
\newtheorem{definition}[theorem]{Definition}

\theoremstyle{remark}
\newtheorem{remark}[theorem]{Remark}

\newcommand{\beq}{\begin{eqnarray}}
\newcommand{\eeq}{\end{eqnarray}}

\newcommand{\ketbra}[2]{\ket{#1}\!\bra{#2}}
\newcommand{\proj}[1]{\ketbra{#1}{#1}}
\newcommand{\Tr}{\mbox{\rm Tr}}
\newcommand{\Id}{\ensuremath{\mathop{\rm Id}\nolimits}}
\newcommand{\Es}[1]{\textsc{E}_{#1}}
\newcommand{\epr}{\mathrm{EPR}}

\newcommand{\EPRd}{\ket{\epr_d}}
\newcommand{\1}{^{-1}}
\newcommand{\dagg}{^{\dagger}}
\renewcommand{\th}{\textsuperscript{th}}
\newcommand{\reg}[1]{{\textsf{#1}}}
\newcommand{\Ext}{\textsc{Ext}}
\newcommand{\mH}{\mathcal{H}}

\newcommand{\C}{\ensuremath{\mathbb{C}}}

\newcommand{\R}{\ensuremath{\mathbb{R}}}
\newcommand{\Z}{\ensuremath{\mathbb{Z}}}

\newcommand{\eps}{\varepsilon}
\newcommand{\s}{\sigma}
\newcommand{\w}{\omega}

\DeclareMathOperator{\cor}{cor}

\newcommand{\setft}[1]{\mathrm{#1}}

\newcommand{\pos}[1]{\setft{Pos}\left(#1\right)}

\newcommand{\pattern}{\ensuremath{\mathcal{P}}}
\newcommand{\spec}{\ensuremath{\mathcal{S}}}
\newcommand{\grid}{\textsc{grid}}
\newcommand{\bx}{\ensuremath{\mathrm{Box}}}
\newcommand{\bad}{\textsc{bad}}
\newcommand{\GHZ}{\textsc{GHZ}}

\newcommand\Prob[1]{\Pr\left( #1 \right)}
\newcommand\floor[1]{\left\lfloor #1 \right\rfloor}

\newcommand{\m}[1]{\mathcal{#1}}
	\newcommand{\G}{\Gamma}
	\newcommand{\seq}{\subseteq}
	
	\renewcommand{\set}[1]{\left\{#1\right\}}
\newcommand{\paren}[1]{\left(#1\right)}
\newcommand{\abs}[1]{\left\vert#1\right\vert}

\newcommand{\GHZk}[1]{\GHZ^{#1}}

\newcommand{\game}{\mathcal{G}}
\newcommand{\ol}[1]{\overline{#1}}
\newcommand{\Hmin}{H_{\text{min}}}

\begin{document}

\title{Trading locality for time: certifiable randomness from low-depth circuits}
\author{
Matthew Coudron 
\thanks{
University of Waterloo, Canada.
\texttt{mcoudron@gmail.com}
}
\and 
Jalex Stark 
\thanks{
University of California Berkeley, USA.
\texttt{jalex@cs.berkeley.edu}
}
\and 
Thomas Vidick
\thanks{
California Institute of Technology, USA.
\texttt{vidick@cms.caltech.edu}
}
}
%}

\maketitle

\begin{abstract}
The generation of certifiable randomness is the most fundamental information-theoretic task that meaningfully separates quantum devices from their classical counterparts. We propose a protocol for exponential certified randomness expansion using a single quantum device. The protocol calls for the device to implement a simple quantum circuit of constant depth on a 2D lattice of qubits. The output of the circuit can be verified classically in linear time, and is guaranteed to contain a polynomial number of certified random bits assuming that the device used to generate the output operated using a (classical or quantum) circuit of sub-logarithmic depth. This assumption contrasts with the locality assumption used for randomness certification based on Bell inequality violation  and more recent proposals for randomness certification based on computational assumptions. Furthermore, to demonstrate randomness generation it is sufficient for a device to sample from the ideal output distribution within constant statistical distance.

Our procedure is inspired by recent work of Bravyi et al. (Science 2018), who introduced a relational problem that can be solved by a constant-depth quantum circuit, but provably cannot be solved by any classical circuit of sub-logarithmic depth. We develop the discovery of Bravyi et al. into a framework for robust randomness expansion.  Our results leads to a new proposal for a demonstrated quantum advantage that has some advantages compared to existing proposals. First, our proposal does not rest on any complexity-theoretic conjectures, but relies on the physical assumption that the adversarial device being tested implements a circuit of sub-logarithmic depth. Second, success on our task can be easily verified in classical linear time. Finally, our task is more noise-tolerant than most other existing proposals that can only tolerate multiplicative error, or require additional conjectures from complexity theory; in contrast, we are able to allow a small constant additive error in total variation distance between the sampled and ideal distributions.  
\end{abstract}

%-------------------------%
\section{Introduction}
\label{sec:introduction}

%!TEX root = main.tex

A fundamental point of departure between quantum mechanics and classical theory is that the former is non-deterministic: quantum mechanics, through the Born rule, posits the existence of experiments that generate \emph{intrinsic randomness}. This observation leads to the simplest and most successful ``test of quantumness'' to have been designed and implemented: the Bell test~\cite{Bell:64a}. Far beyond its role as a test of the foundations of quantum mechanics, the Bell test has become a fundamental building block in quantum information, from protocols for quantum cryptography (e.g. device-independent quantum key distribution~\cite{Eke91,VV13prl}) to complexity theory (e.g. delegated quantum computation~\cite{ReichardtUV13nature}, multiprover interactive proof systems~\cite{CHTW04}) and much more~\cite{brunner2014bell}. 
Yet, while a loophole-free implementation of a Bell test has been demonstrated~\cite{hensen2015experimental,giustina2015significant,shalm2015strong} it remains a challenging experimental feat, which unfortunately leaves its promising applications wanting (here "loophole-free" refers to a stringent set of experimental standards which ensure that all required assumptions have been verified ``beyond reasonable doubt''). The increasingly powerful quantum devices that are being experimentally realized tend to be single-chip, and do not have the ability to implement loophole-free Bell tests. The task of devising convincing ``tests of quantumness'' for such devices is challenging. 

Until recently the only proposal for such tests was the design of so-called ``quantum supremacy experiments''~\cite{harrow2017quantum}, which specify classical sampling tasks that can in principle be implemented by a mid-scale quantum device, but cannot be simulated by any efficient classical randomized algorithm (under somewhat standard computational assumptions~\cite{aaronson2017complexity,harrow2017quantum}). These proposals share a number of well-recognized limitations. Firstly, while the sampling part of the process can be done efficiently on a quantum computer, verifying that the quantum computer is sampling from a hard distribution requires a computational effort which scales exponentially in the number of qubits. Secondly, their experimental realization is hindered by a generally poor tolerance to errors in the implementation, which is compounded by the necessity to implement circuits with relatively large (say, at least $\sqrt{N}$ for an $N\times N$ grid) depth. Combined with the resort to complexity-theoretic assumptions for which there is little guidance in terms of concrete parameter settings (see however~\cite{dalzell2018many}), this has led to an ongoing race in efficient simulations~\cite{chen201864,huang2018explicit,markov2018quantum}. Indeed, the proposals operate in a limited computational regime, requiring a machine with, say, at least 50 qubits (to prevent direct clasical simulation) but at most 70 qubits (so that verification can be performed in a reasonable amount of time) --- leaving open the question of what to do with a device with more than, say, 100 qubits. At a more conceptual level, the proposals are based on computational tasks that appear arbitrary (such as the implementation of a random quantum circuit from a certain class). In particular, they do not lead to any further characterization of the successful device, that could be used to e.g. build a secure delegation protocol or even simply certify a simple property such as the preparation of a specific quantum state or the implementation of a certain class of measurements. 

\medskip

We propose a different kind of experiment, or ``test of quantumness'', for large but noisy quantum devices, that is inspired from recent work of Bravyi et al. on the power of low-depth quantum circuits~\cite{bravyi2017quantum}.
Our test is applicable in a regime where the device has a large number of qubits, but may only have the ability to implement circuits of low (constant) depth, due e.g. to a limited gate fidelity. We argue that the test overcomes the main limitations outlined above: it generates useful outcomes (certifiably random bits), it is easily verifiable (in classical linear time), and it is robust to a small amount of error (it is sufficient to generate outputs within constant statistical distance from the ideal distribution\footnote{In fact, even less is needed; see the description of the protocol.}). 
The test does not require any assumption from complexity theory, but instead considers a novel physical assumption (introduced in~\cite{bravyi2017quantum}): that the device implements a circuit whose depth is at most a small constant times the logarithm of the number of qubits. Intuitively, this assumption trades off locality (as required by the Bell test) for time (as measured by circuit depth). It is particularly well-suited to quantum devices for which the number of qubits can be made quite large, but the gate fidelity remains low, limiting the depth of a circuit that can be implemented. Informally, we show the following.

\begin{theorem}\label{thm:main-informal}
Let $N$ be an arbitrary integer and $M=N^2$. There exists universl constants $c,d>0$, a distribution $\mathcal{D}$ on $\{0,1\}^{M}$, and an efficiently verifiable relation $\mathcal{R}\subseteq \{0,1\}^M \times \{0,1\}^M$ such that the followinog holds. Let $\m C$ be a (classical or quantum) circuit with gates of constant fan-in and depth $D \leq c \log M$ such that the output of the circuit satisfies the relation non-negligibly often, i.e.
\[\Pr_{x\sim \mathcal{D}}\big[(x,\mathcal{C}(x))\in \mathcal{R}\big]\,=\,\Omega\big(M^{-1/10}\big)\;.\]
Then $\m C$ achieves \emph{exponential randomness expansion}: a sample from $\mathcal{D}$ can be obtained using $O(\log^2 M)$ uniformly random bits and outputs of the circuit must have $\Omega(M^{d})$ bits of entropy.
\end{theorem}

We refer to Theorem~\ref{thm:main-randomness} for a precise statement. In particular, the output entropy is quantified using the quantum conditional min-entropy, conditioned on the inputs to the circuit and quantum side information that may be correlated with the initial state of the circuit. %Our main result is even stronger: it implies that even a classical randomized circuit of low depth is unable to generate outputs that satisfy the relation $\m R$ with all but negligible probability, thereby fully justifying the label of ``test of quantumness''.
 Note that the resulting test is far more robust than most existing proposals, that require the output distribution to be multiplicatively close to the target distribution. In contrast, in our case it is sufficient to hit a certain target (the relation $\mathcal{R}$, that itself is very permissive) an inverse polynomial fraction of the time!

Aside from the application to randomness expansion, Theorem~\ref{thm:main-informal} strengthens the main result of Bravyi et al.~\cite{bravyi2017quantum} in multiple ways. Bravyi et al. provide a relation such that for any classical circuit of sufficiently low depth, there \emph{exists} an input such that the circuit must return an output that satisfies the relation with probability bounded away from $1$. In contrast, we point out the existence of an efficiently sampleable distribution on inputs such that, for any classical low-depth circuit, we know that \emph{on average} over the choice of an input the circuit returns an output that satisfies the relation with at most small probability.  While this improvement follows using a simple extension of the arguments in~\cite{bravyi2017quantum}, it is key to the practical relevance of the scheme. In addition we make a further improvement and address the following question left open in~\cite{bravyi2017quantum}: how small can the maximum success probability of all low-depth classical circuits (i.e. the ``soundness") be made? 

\begin{theorem}[Exponential soundness]\label{cor:exp-bgk-expsoundess}
	Let $N$ be an arbitrary integer and $M=N^2$. There exists a distribution $\mathcal{D}$ on $\{0,1\}^{M}$ and an efficiently verifiable relation $\mathcal{R}\subseteq \{0,1\}^M \times \{0,1\}^M$ such that the following holds, for universal constants $c,c'>0$: 
	\vspace{-0.2cm}
	\begin{itemize} \itemsep0em
		\item (Completeness) There exists a depth-$3$ geometrically local (in 2D) quantum circuit such that for any input $x$ in the support of $\mathcal{D}$ the circuit samples a $y$ such that $(x,y)\in \mathcal{R}$ with probability $1$.
		\item (Soundness) For any classical circuit of depth $D\leq c\log N$, the probability that $(x,y)\in \mathcal{R}$, for $x\sim\mathcal{D}$ and $y$ the output of the circuit on input $x$, is $O(\exp(-N^{c'}))$.
	\end{itemize}
\end{theorem}
	\vspace{-0.1cm}
Note that the improvement in soundness between our two results is enabled by the fact that in Theorem~\ref{cor:exp-bgk-expsoundess} it is no longer the case that it is possible to sample from $\mathcal{D}$ using poly-logarithmically many bits.     Arguably, good soundness guarantees are crucial to a successful experimental demonstration: due to the presence of noise the quantum device cannot be expected to succeed with probability arbitrarily close to $1$, so that the lower the performance of classical circuits, the lower the requirements on the quantum circuit as well.  

\paragraph{Discussion.}
We comment on the depth assumption that underlies our results, and their potential for a practical demonstration of a quantum advantage (a.k.a. ``quantum supremacy experiment''). The quantum circuit required for a successful implementation of our task is relatively straightforward to implement. It can be realized in three phases. A first, offline phase initializes EPR pairs (or three-qubit GHZ states) between nearest-neighbor qubits on a 2D grid. In a second phase, each qubit is provided an input, according to which either the qubit should be measured according to a single-qubit Pauli observable, or the qubit and one of its neighbors should be measured in the Bell basis. Finally, in the third phase the measurement outcomes are aggregated and verified using a simple classical linear-time computation. 

In order to demonstrate a quantum advantage, the crucial requirement is that the second phase should be implemented using a procedure that is ``certified'' to have low depth. Since this is a physical assumption, it can never be rigorously proven. Nevertheless, it is possible to imagine experiments under which the assumption would hold ``beyond reasonable doubt''. We describe two such experiments. 

In a first scenario, the verification of the depth constraint could be based on a calculation that takes into account state-of-the-art clock speeds. The fastest classical processors operate at speeds of order $1$GHz, so that for an integer $N$, a circuit of depth $d=\log(N)$ takes time of order $10^{-9}\log(N)$ seconds to implement. In contrast, current gate times for, say, ion-trap quantum computers are of order $100$ nanoseconds~\cite{schafer2018fast}, meaning that the quantum circuit realizing our task could be implemented in time roughly $10^{-7}$ seconds. To observe a quantum advantage it is thus necessary to ensure $\log(N) \gg 10^2$, leading to an impractical circuit size. However, a reasonable factor $10$ improvement in the gate time for quantum gates could enable a demonstration based on a grid of order $2^{10}\times 2^{10}$ qubits. Although far beyond current capabilities, the number is not beyond reach. Keeping in mind the extreme simplicity of the task to be implemented, it is not unreasonable to hope that such circuits may exist within 5-10 years.  

In the previous scenario we allowed both the classical and quantum procedures solving our task to do so in a highly localized, single-chip fashion. The distributed nature of the task lends itself well to another type of implementation, that would be more demanding for a classical adversarial behavior, and may thus lead to a more practical demonstration of quantum advantage. Consider a network of constant-qubit devices arranged in a $N\times N$ grid, such that devices may be separated by large (say, kilometric) distances. In the first, offline phase the devices use nearest-neighbor  quantum communication channels to distribute EPR pairs. In the second phase, each device receives a classical input, performs a simple local measurement, and returns a classical output (no communication is required). Our result implies that, to even approximately reproduce the output distribution implemented by this procedure, a classical network would need to operate in at least $\Omega(\log N)$ rounds, where in each round a device can communicate with a constant number of devices located at arbitrary locations in the network (the network need not be 2D: at each step, a device is allowed to broadcast arbitrarily but can only receive information from a constant number of devices, whose identity must be fixed ahead of time). Taking into account inevitable latency delays incurred in any such network, this second scenario suggests that our task may lead to an interesting test for a future quantum internet~\cite{wehner2018quantum}. 

Finally we comment on the fidelity requirement for the gates of a quantum circuit implementing our task. Even though the circuit is only of constant depth, it is important that, along a typical path of length $O(N)$ between two qubits in the $N\times N$ grid, none of the gates leads to an error. This means that per-gate fidelity is required to be of order $1-O(1/N)$. For $N$ of order $2^{10}$, as suggested in the first scenario described above, such fidelities are within reach. We also note that by changing the architecture of the circuit from a 2D grid to a 3D grid it may be possible to leverage existing protocols for entanglement distribution using noisy resources~\cite{raussendorf2005long}. Unfortunately, this comes with the drawback of a challenging 3D architecture for which there is no current implementation. 

\vspace{-0.2cm}

\paragraph{Proof idea.}
Our starting point is the key observation, made by Bravyi et al.~\cite{bravyi2017quantum}, that a sub-logarithmic depth circuit made of gates with constant fan-in has a form of implied locality, where the ``forward lightcone'' of most input vertices only includes a vanishing fraction of output vertices. In particular, two randomly chosen input locations are unlikely to have overlapping lightcones. If the input to the circuit is non-trivial in those two locations only, then the outputs in each input location's forward lightcone are obtained by a computation that depends on that input only. In other words, we have a reduction from classical, low-depth circuits to  two-party local computation that exactly preserves properties of the output. 
While the same lightcone argument holds true for a quantum circuit, the quantum circuit has the ability of distributing entanglement across any two locations in depth $2$, by executing a sequence of entanglement swapping procedures in parallel. Thus the same reduction maps a quantum, low-depth circuit to a two-party local computation, where the parties may perform their local computation on a shared entangled state. Since there are well-known separations between the kinds of distributions that can be generated by performing local operations on an entangled state, as opposed to no entanglement at all --- this is precisely the scope of Bell inequalities --- Bravyi et al. have obtained a separation between the power of low-depth classical and quantum circuits. 

We build on this argument in the following way. Our first contribution is to boost the argument in~\cite{bravyi2017quantum} from a worst-case to a ``high probability'' statement. Instead of showing that (i) for every classical circuit, there is some choice of input on which the classical circuit will fail, and (ii) there is a quantum circuit that succeeds on every input, we show that there exists a suitable distribution on inputs that is such that, (i) any classical circuit fails with high probability given an input from the distribution, and (ii) there is a quantum circuit that succeeds with high probability (in fact, probability $1$) on the distribution. Second, we observe that the construction in~\cite{bravyi2017quantum} imposes constraints not only on classical low-depth circuits, but also on quantum low-depth circuits; this observation enables the reduction to nonlocal games hinted at above. Finally, we amplify the argument to show how a polynomial number of Bell experiments can be simultaneously ``planted'' into the input to the circuit. This allows us to perform a reduction to a nonlocal game in which there is a large number of players divided into pairs which each perform their own distinct Bell experiment.  By adapting techniques from the area of randomness expansion from nonlocal games~\cite{arnon2018practical} we are then able to conclude that any sub-logarithmic-depth circuit, classical or quantum, that succeeds on our input distribution, must generate large amounts of entropy. Moreover, this guarantee holds even if the circuit only correctly computes a sufficiently large but constant fraction of outputs for the games. %This allows us to obtain an exponential soundness statement, while at the same time allowing for a significant amount of noise in a quantum circuit before a ``quantum advantage'' is certified. 

\paragraph{Related work.} Two recent works investigate the question of certified randomness generation outside of the traditional framework of Bell inequalities. In~\cite{brakerski2018certifiable} randomness is guaranteed based on the computational assumption that the device does not have sufficient power to break the security of post-quantum cryptography. The main advantages of this proposal are that the assumption is a standard cryptographic assumption, and that verification is very efficient. A drawback is the interactive nature of the protocol, where only a fraction of a bit of randomness is extracted in each round. In~\cite{aaronson18}, Aaronson announced a randomness certification proposal based on the Boson Sampling task. The main advantage of the proposal is that it can potentially be implemented on a device with fewer than $100$ qubits. Drawbacks are the difficulty of verification, that scales exponentially, and the resort to somewhat non-standard  complexity conjectures, for which there is little evidence of practical hardness (e.g. it may not be clear how to set parameters for the scheme so that an adversarial attack would require time $2^{80}$). In comparison, we would say that an advantage of our proposal is its simplicity to implement (on an axis different from Aaronson's: we require many more qubits, but a much simpler circuit, of constant depth and with classically controlled Clifford gates only), its robustness to errors, and its ease of verification. A possible drawback is the physical  assumption of bounded depth, that may or may not be reasonable depending on the scenario (in contrast to cryptographic or even complexity-theoretic assumptions, that operate at a higher level of generality).

Two other works obtained concurrently and independently from ours establish directly related, but strictly incomparable, results. In~\cite{gall2018average} Le Gall obtains an average-case hardness result that is very similar to our Theorem~\ref{cor:exp-bgk-expsoundess}, with a concrete constant $c' = 1/2$ that is likely better than the one that we achieve here. Le Gall's proof is based on an ingenious construction using the framework of graph states; although some aspects are similar in spirit to ours (such as the use of parallel repetition to amplify the soundness guarantees) the proof rests on rather different intuition. In independent work, Bene Watts et al.~\cite{BWKST18} extend the results of~\cite{bravyi2017quantum} to obtain a result analogous to our Theorem~\ref{cor:exp-bgk-expsoundess}, with a strengthened soundness property which holds even against so-called $AC^0$ circuits.  $AC^0$ circuits are still required to have constant depth but may contain AND and OR gates of arbitrary fan-in (instead of constant fan-in for~\cite{bravyi2017quantum} and our results). Their proof applies to the same relation as~\cite{bravyi2017quantum} but uses more involved techniques from classical complexity theory to obtain the strengthened lower bound. Neither of these results obtains an application to randomness expansion as in our Theorem~\ref{thm:main-informal}.

\paragraph{Acknowledgments.} 
The authors thank Adam Bouland for helpful discussions and members of the Caltech theory reading group (Matthew Weidner, Andrea Coladangelo, Jenish Mehta, Chinmay Nirkhe, Rohit Gurjar, Spencer Gordon) for posing some of the questions answered in this work. We thank Isaac Kim, Jean-Francois Le Gall, and Robin Kothari for useful discussions following the initial announcement of our results. 

Thomas Vidick is supported by NSF CAREER Grant CCF-1553477, AFOSR YIP award number FA9550-16-1-0495, MURI Grant FA9550-18-1-0161, a CIFAR Azrieli Global Scholar award, and the the Institute for Quantum Information and Matter, an NSF Physics Frontiers Center (NSF Grant PHY-1733907).
Jalex Stark is supported by NSF CAREER Grant CCF-1553477, ARO Grant W911NF-12-1-0541, and NSF Grant CCF-1410022.
Matthew Coudron is supported by Canada's NSERC and the Canadian Institute for Advanced Research (CIFAR), and through funding provided to 
IQC by the Government of Canada and the Province of Ontario.

%-------------------------%
%\section{Conjectures}
%\label{sec:conjectures}

%\input{conjectures}

%-------------------------%
\section{Preliminaries}
\label{sec:prelim}

%!TEX root = main.tex

\subsection{Notation}\label{sec:notation}

Finite-dimensional Hilbert spaces are designated using caligraphic letters, such as $\m H$. A \emph{register} $\reg{A}$, $\reg{B}$, $\reg{R}$, represents a physical subsystem, whose associated Hilbert space is denoted $\mH_{\reg{A}}$, $\mH_\reg{B}$, etc.
 We write $\Id_\reg{R}$ for the identity operator on $\mH_\reg{R}$.
%We write $\pos{\mathcal{H}}$ for the positive semidefinite operators on $\m H$, and $\unitary{\mH}$ for the unitary operators. 
A POVM $\{M^a\}$ on $\m H$ is a collection of positive semidefinite operators on $\m H$ such tht $\sum_a M^a = \Id$. For $X$ a linear operator on $\m H$, we write $\Tr(X)$ for the trace and $\|X\|_1 =\Tr\sqrt{X^\dagger X}$ for the Schatten-$1$ norm. 

For an integer $d\geq 1$ an \emph{observable over $\Z_d$} is a unitary operator $A$ such that $A^d=\Id$. For $\omega = e^{\frac{2i\pi}{d}}$ and taking addition modulo $d$ we write 
$$X = \sum_{i=0}^{d-1} \ket{i+1}\!\bra{i}\;\qquad\text{and}\qquad Z = \sum_{i=0}^{d-1} \omega^i \ket{i}\!\bra{i}\;$$
for the generalized qudit Pauli $X$ and $Z$ operators, which are observables acting on $\mH=\C^d$.
	Given an integer $d\geq 1$ and a tuple $s \in \Z_d^2$, we write $\sigma_s = X^{s_0}Z^{s_1}$ for a one-qudit Pauli acting on $\C^d$. 
	Given an integer $n\geq 1$ and a string $r \in (\Z_d^2)^n$, we write $\sigma_{r} = \otimes_i \sigma_{r_i}$ for an $n$-qudit Pauli acting on $(\C^d)^{\otimes n}$.

\subsection{Nonlocal games}

We consider two types of games: multiplayer nonlocal games, and circuit games. Circuit games are nonstandard, and we introduce them in Section~\ref{sec:circuit-games}. Nonlocal games are defined as follows.

\begin{definition}[Nonlocal game]\label{def:nl-game}
 Let $\ell\geq 1$ be an integer. An \emph{$\ell$-player nonlocal game} $\game$ consists of finite question and answer sets $X=X_1\times\cdots\times X_\ell$
%!TEX root = main.tex
  and $A=A_1\times\cdots\times A_\ell$ respectively, a distribution $\pi$ on $X$, and a family of coefficients $V(a_1,\ldots,a_\ell|x_1,\ldots,x_\ell)\in [0,1]$, for $(x_1,\ldots,x_\ell)\in X$ and $(a_1,\ldots,a_\ell)\in A$. We call an element $x\in X$ in the support of $\pi$ a \emph{query}, and for $i\in\{1,\ldots,\ell\}$ the $i$-th entry $x_i$ of $x$ a \emph{question} to the $i$-th player. We refer to the function $V(\cdot|\cdot)$ as the win condition for the game, and for any query $x$, to a tuple $a$ such that $V(a|x)=1$ as a valid (or winning) tuple of answers (to query $x$). When players return valid answers we say that they win the game.  
	\end{definition}
	
	\begin{definition}[Strategy]\label{def:meas-strategy} 
 Let $\ell\geq 1$ be an integer, and $\game$ an $\ell$-player nonlocal game. An \emph{$\ell$-player strategy} $\tau = (\rho, \{M_{x_i}\})$ for $\game$ consists of an $\ell$-partite
    density matrix $\rho \in \mathcal{H}_1\otimes \cdots \otimes\mathcal{H}_\ell$, and for each $i\in\{1,\ldots,\ell\}$ a collection of measurement
  operators $\{M^{a_i}_{x_i}\}_{a_i \in A_i}$ on $\mathcal{H}_i$ indexed by $x_i \in X_i$ and with outcomes $a_i\in A_i$.
\end{definition}
	
	\begin{definition}[Game value]
	Let $\game$ be an $\ell$-player nonlocal game, and  $\tau = (\rho, \{M_{x_i}^{a_i}\})$ a strategy for the players in $\game$. The \emph{value} of $\tau$ in $\game$ is
	\[ \omega_\tau^*(\game)\,=\, \sum_{x_1,\ldots,x_\ell} \pi(x_1,\ldots,x_\ell) \sum_{a_1,\ldots,a_\ell} V(a_1,\ldots,a_\ell| x_1,\ldots,x_\ell) \Tr\big( (M_{x_1}^{a_1} \otimes \cdots \otimes M_{x_\ell}^{a_\ell} )\,\rho\big) \;.\]
	A strategy $\tau$ is called \emph{perfect} if $\omega^*_\tau(\game)=1$. 
	The \emph{entangled value} (or simply \emph{value}) of $\game$, $\omega^*(\game)$, is defined as the supremum over all strategies $\tau$ of $\omega_\tau^*(\game)$.
\end{definition}

To compare strategies
we first introduce a notion of distance between measurements, with respect to an underlying state. (This is a standard definition in the area of self-testing.) 

\begin{definition}[State-dependent distance]
	Let $\rho$ be a density matrix in $\mH$ and let $M = \{M^a\}_a,N = \{N^a\}_a$ be two POVM on $\mH$ that have the same set of outcomes. The \emph{state-dependent distance} between $M$ and $N$ is 
	\begin{align}
		d_\rho(M,N) \,=\, \Big( \sum_a \Tr \big( (M^a - N^a)^2 \rho \big) \Big)^{1/2}\;.
	\end{align}
\end{definition}

\begin{definition}[Closeness of strategies]
\label{def:close}
	Let $\tau = (\rho,\{M_{x_i}^{a_i}\})$, $\tilde{\tau} = (\tilde{\rho},\{\tilde{M}_{x_i}^{a_i}\})$ be strategies for an $\ell$-player nonlocal game $\game$. We say that $\tau$ is \emph{$\eps$-close} to $\tilde{\tau}$ if and only if
 $ \| \rho - \tilde{\rho}\|_{1} \leq \eps$
		and for all $i\in\{1,\ldots,\ell\}$ it holds that $\Es{x} d_\rho(M_{x_i},\tilde{M}_{x_i}) \leq \eps$, where the expectation is over $x=(x_1,\ldots,x_\ell)$ drawn from $\pi$.
\end{definition}

\begin{definition}[Isometric strategies]
	Let $\tau = (\rho,\{M_{x_i}^{a_i}\})$ and $\tau' = (\rho',\{(M')_{x_i}^{a_i}\})$ be strategies for an $\ell$-player nonlocal game $\game$, and $\eps>0$. We say that $\tau$ is \emph{$\eps$-isometric} to $\tau'$ if and only if there exist isometries $V_i : \mathcal{H}_i \to \mathcal{H}_i'$ for each $i\in\{1,\ldots,\ell\}$ such that $\tau'$ is $\eps$-close to the strategy $\tilde{\tau} = (\tilde{\rho}, \{ \tilde{M}_{x_i}^{a_i} \})$, where 
$\tilde{\rho} = (V_1 \otimes \cdots \otimes V_\ell) \rho (V_1 \otimes \cdots \otimes V_\ell)^\dagger$
		and for all $i\in\{1,\ldots,\ell\}$, $x_i\in X_i$ and $a_i\in A_i$, $\tilde{M}_{x_i}^{a_i} = V_i  M_{x_i}^{a_i} V_i^\dagger$.
		\end{definition}

\begin{definition}\label{def:robust-rigid}
We say that a game $\game$ is \emph{robustly rigid} if the following holds. There is a continuous function $f:\R_+ \to \R_+$ such that $f(0)=0$ and a strategy $\tau$ for $\game$ such that for any $\delta\geq 0$, any strategy $\tau'$ with value at least $\omega_\tau^*(\game)-\delta$ is $f(\delta)$-isometric to $\tau$. We refer to $f$ as the robustness of the game. 
\end{definition}

Note that for a game to be robustly rigid it is necessary that there exists a unique strategy $\tau$ such that $\omega^*_\tau(\game)=\omega^*(\game)$, up to isometry.

\subsection{Circuits}

We refer to~\cite{nielsen2002quantum} for an introduction to the quantum circuit model. 
We consider layered circuits over an arbitrary gate set. The choice of a specific gate set may affect the depth of a circuit; for concreteness, the reader may consider the standard gate set $\{X,Z,H,T,CNOT\}$, where here $X,Z$ are the Pauli observables over $\C^2$, 
\[ H = \frac{1}{\sqrt{2}}\begin{pmatrix} 1 & 1 \\ 1 & -1 \end{pmatrix}\;,\qquad T = \begin{pmatrix} 1 & 0 \\ 0 & e^{i\pi/4} \end{pmatrix}\;,\]
and $CNOT$ is the controlled-NOT gate. In general, gates in the gate set used to specify the circuit may have arbitrary fan-out, but are restricted to fan-in at most $K$, where $K\geq 2$ is a parameter that is considered a constant (in contrast to the depth $D$ of the circuit, that is allowed to grow with the number of input wires to the circuit). Note that if $\m C$ is a quantum circuit, ``fan-in'' is the same as  locality, i.e. the number of qubits that a gate acts on nontrivially. In particular, for quantum circuits bounded fan-in automatically implies bounded fan-out.

It is convenient to generalize the usual notion of Boolean circuit to allow circuits that act on inputs taken from a larger domain, e.g. ${\m C}: \Sigma^n \to \Sigma^m$, where $\Sigma$ is a finite alphabet. Similar to the fan-in, whenever using the $O(\cdot)$ notation we consider the cardinality of $\Sigma$ a constant. A circuit of depth $D$ and fan-in $K$ over $\Sigma$ can be converted in a straightforward way in a circuit of depth $D$ and fan-in $K\cdot\lceil\log_2|\Sigma|\rceil$ over $\{0,1\}$. For the case of quantum circuits, allowing a non-Boolean $\Sigma$ amounts to considering a circuit that operates on $d$-dimensional qudits, for $d=|\Sigma|$, instead of $2$-dimensional qubits. 

\subsection{Entropies}

Given a bipartite density matrix $\rho_{\reg{AB}}$ we write $H(A|B)_\rho$, or simply $H(A|B)$ when $\rho_{\reg{AB}}$ is clear from context, for the conditional von Neumann entropy, $H(A|B)=H(AB)-H(B)$, with $H(X)_\sigma = - \Tr(\sigma\ln\sigma)$ for any density $\sigma$ on $\mH_\reg{X}$. 
We recall the definition of (smooth) min-entropy.

\begin{definition}[Min-entropy]\label{def:min-entropy}
Let $\rho_\reg{XE}$ be a density matrix on two registers $\reg{X}$ and $\reg{E}$, such that the register $\reg{X}$ is classical. The \emph{min-entropy} of $\reg{X}$ conditioned on $\reg{E}$ is defined as
  \begin{equation*}
    \Hmin({X|E})_\rho = \max \{\lambda \geq 0 :  \exists \sigma_\reg{E} \in \pos{\mathcal{H}_\reg{E}}, \Tr(\sigma_\reg{E} )\leq1, \, \mathrm{s.t.}\,\, 2^{-\lambda} \Id_\reg{X} \otimes \sigma_\reg{E} \geq \rho_\reg{XE}\}.
  \end{equation*}
When the state $\rho$ with respect to which the entropy is measured is clear from context we simply write $\Hmin({X|E})$ for $\Hmin({X|E})_\rho$.	
  For $\eps\geq 0$ the
  \emph{$\eps$-smooth min-entropy} of $X$ conditioned on $E$ is defined as
  \begin{equation*}
    \Hmin^\eps(X|E)_\rho = \max_{\sigma_\reg{XE} \in \mathcal{B}(
      \rho_\reg{XE},\eps) } \Hmin(X|E)_\sigma,
  \end{equation*}
	where $\mathcal{B}(
      \rho_\reg{XE},\eps) $ is the ball of radius $\eps$ around $\rho_\reg{XE}$, taken with respect to the purified distance.\footnote{The definition of the purified distance is not important for us, and we defer to~\cite{tomamichel2015quantum} for a precise definition.}
\end{definition}

The following theorem justifies the use of the smooth min-entropy as the appropriate notion of entropy for randomness extraction. 

\begin{theorem}[\cite{de2012trevisan}]\label{thm:qext}
For any integers $n,m$ and for any $\eps > 0$ there exists a $d = O(\log^2(n/\eps) \cdot \log{m})$ and an efficient classical procedure  $\Ext: \{0,1\}^n \times \{0,1\}^d \to \{0,1\}^m$ such that for any density matrix $\rho_{\reg{XE}} = \sum_x \proj{x}_{\reg{X}} \otimes \rho_{\reg{E}}^x$ such that the register $X$ is an $n$-bit classical register and $\Hmin(X|E)\geq 2m$, letting $\rho_\reg{ZYE} = 2^{-d} \sum_{x,y} \proj{\Ext(x,y)}_\reg{Z}\otimes \proj{y}_\reg{Y} \otimes \rho_\reg{E}^x$ it holds that 
\[ \big\| \rho_\reg{ZYE} - U_m \otimes U_d \otimes \rho_\reg{E} \big\|_1 \,\leq\,\eps\;,\]
where for an integer $\ell \geq 1$, $U_\ell = 2^{-\ell} \Id$ is the totally mixed state on $\ell$ qubits and $\rho_\reg{E} = \sum_x \rho_\reg{E}^x$. 
\end{theorem}

%-------------------------%
\section{Stabilizer games}
\label{sec:stabilizer-games}

%!TEX root = main.tex

In this section we introduce a restricted class of nonlocal games that we will be concerned with throughout the paper. We call the games \emph{stabilizer games}. They have the property that the game always has a perfect quantum strategy $\tau=(\rho,\{M_{x_i}\})$ that uses an entangled state $\rho = \proj{\psi}$ that is a graph state, on which the players make measurements that are specified by tensor products of Pauli observables. It is important for our results that there is a perfect strategy such that the entangled state can be prepared by a quantum circuit of low depth (in fact, constant depth) starting on a $\ket{0}$ state. It will also be convenient that the same perfect strategy only requires the measurement of Pauli operators, and that the win condition in the game is a linear function of the players' answers. 

We proceed with a formal definition. The games we consider have $\ell$ players. In the intended strategy for the players, each player $j\in\{1,\ldots,\ell\}$ holds $k_j$ qudits, measures $m_j$ commuting Pauli observables over $\Z_d$ (depending on its question), and reports the $m_j$ outcomes. 

\begin{definition}[Stabilizer game]\label{def:stab-game}
	An $(\ell,k,m)$ stabilizer game $\game = (X_i,\{w_x,b_x\})$ is an $\ell$-player nonlocal game defined from the following data. 
	\begin{itemize}
			\item a number of players $\ell$,
					\item a parameter $d$ for the dimension of the qudits (in the honest strategy),
		\item for $j\in\{1,\ldots,\ell\}$, a parameter $k_j$ for the number of qudits held by the $j$-th player (ibid),
		\item for $j\in\{1,\ldots,\ell\}$, a parameter $m_j$ for the number of simultaneous measurements made by the $j$-th player (ibid),
		\item for $j \in \{1,\ldots,\ell\}$, a set $X_j$, each element of which is identified with the label $x\in(\Z_d^2)^{k_j}$ of a $k_j$-qudit Pauli,
			\item a distribution $\pi$ on queries $x \in \prod_{j=1}^\ell X_j^{m_j}$, such that any $(x_1,\ldots,x_\ell)$ in the support of $\pi$ is such that for each $j$, $x_j$ designates an $m$-tuple of commuting $k_j$-qudit Pauli observables, 
			\item for each query $x$ in the support of $\pi$, a vector $ w_x \in \prod_{j=1}^\ell (\Z_d)^{ m_j}$ and a coefficient $b_x \in \Z_d$ that are used to specify the win condition in the game.
	\end{itemize}
	To play, the verifier samples a question $x_j \in X_j^m$ for each player. 
	Each player responds with a string $a_j \in \Z_d^{m_j}$.
	Let $x = (x_1,\ldots, x_{\ell})$ and $a = (a_1,\ldots, a_{\ell})$. The players win if 
	\begin{equation}\label{eq:stabilizer-win-condition}
		 w_x \cdot a \,=\, b_x\;,
	\end{equation}
	where the inner product is over vectors in $\Z_d^{\sum m_j}$. 
	Using the notation from Definition~\ref{def:nl-game}, $V(a|x) = 1$ if $w_x \cdot a = b_x$, and $0$ otherwise. 
\end{definition}

% \begin{figure}
%\begin{mdframed}[backgroundcolor=black!15]
	%\jnote{This figure is purely for exposition. I think we should consolidate what's currently written as Definitions \ref{def:rotated-magic-square} and \ref{def:stretched-magic-square} into this figure. The formatting of this figure is also not final}
	%The Magic square game is a stabilizer game with $m = k = l = 2$. Alice answers with the first two elements in her row or column. The answer coefficient $b_q$ is always $0$.  The answer vector is always one of $(1,0,1,0), (0,1,1,0), (1,1,1,0)$, depending on whether Bob's question matches the first, second, or third element of Alice's line, respectively. 
%\end{mdframed}
% We show explicitly how to realize the magic square game as a stabilizer game, and we give some properties of the stretched magic square game.}
% \end{figure}

In a stabilizer game each player is tasked with reporting $m$ values in $\Z_d$. It is then natural to use a representation of strategies in terms of observables over $\Z_d$. We adapt Definition~\ref{def:meas-strategy} as follows. 

\begin{definition}
Let $\game = (X_i, \{w_x,b_x\})$ be a stabilizer game. A \emph{strategy} $\tau = (\rho,\{M_{x_j}\})$ for $\game$ is specified by an $\ell$-partite density matrix $\rho$ and for each $j\in\{1,\ldots,\ell\}$ and $x_j = (x_{j,1},\ldots,x_{j,m_j}) \in X_j^{m_j}$ a family of $m_j$-tuples of commuting observables $M_{x_j} = (M_{x_j,1},\ldots,M_{x_j,m_j})$ over $\Z_d$. 
\end{definition}

Note that in the definition, for $s\in\{1,\ldots,m_j\}$ the observable $M_{x_j,s}$ may depend on the whole $m_j$-tuple $x_j$, and not only on $x_{j,s}$. 

We introduce a notion of ``honest strategy'' in a stabilizer game.

\begin{definition}[Honest strategy]\label{def:honest-stabilizer}
Let $\game = (X_j, \{w_x,b_x\})$ be a stabilizer game. A \emph{honest strategy} in $\game$ is a strategy in which the state $\rho$ is an $(\sum_j k_j)$-qudit $\ell$-partite pure state $\ket{\psi}$ such that the $j$-th player has $k_j$ qubits, and the player's observables $M_{x_j} = (M_{x_j,1},\ldots,M_{x_j,m})$ associated with question $x_j = (x_{j,1},\ldots,x_{j,m_j})\in X_j^{m_j}$ are precisely the $m_j$ commuting Pauli observables specified by $x_j$. We say that the strategy has \emph{depth $d$} if the state $\ket{\psi}$ can be prepared by a quantum circuit of depth at most $d$ starting from the $\ket{0}$ state. 
\end{definition}

\subsection{Pauli observables}

Recall the notation $\sigma_r$, where $r\in(\Z_d^2)^k$, introduced in Section~\ref{sec:notation} to designate an arbitrary $k$-qudit Pauli observable. 

\begin{definition}[Correction value]
Let $q,r \in (\Z_d^2)^k$. The 
\emph{correction value} $\cor_r(q) \in \Z_d$ is defined such that
	\begin{equation}
		\w^{\cor_r(q)} \,=\, [\sigma_q,\sigma_{r}]\;,
	\end{equation}
	where the brackets denote the group commutator, $[P,Q]=PQP^{-1}Q^{-1}$. 
\end{definition}

The following lemma shows that the function $\cor$ can be computed locally. 

\begin{lemma}[$\cor$ can be computed locally]
	\label{lem:cor-computed-locally}
	For a string $s$, let $s|_i$ denote the string which is equal to $s_i$ in the $i\th$ position and $0$ everywhere else. Then
	\begin{equation}
		\sum_i \cor_{r|_i}(q|_i) = \cor_r(q)\;.
	\end{equation}
\end{lemma}

\begin{proof}
	First, notice that 
	\begin{equation}
	\label{eq:cor-computed-locally-1}
		\cor_{r|_i}(q|_i) = \cor_{r|_i}(q)\;.
	\end{equation}
	 To see this, recall that $\cor$ is computed as the phase of the group commutator of a $P_{r|_i}$ and $\s(q)$. We can evaluate this group commutator one tensor factor at a time. In all tensor factors other than $i$, the commutator will be trivial since the $r$ operator is identity. Therefore, the commutator does not change if we also set the $q$ operator to identity.

	Next, we need that for any fixed $q$, the map $r \mapsto \cor_r(q)$ is an additive homomorphism. In other words, 
	\begin{equation}
	\label{eq:cor-computed-locally-2}
		\cor_{r+r'}(q) = \cor_{r}(q)+\cor_{r'}(q).
	\end{equation}
	To see this, we apply Lemma \ref{lemma:commutators} with $A = \sigma_q$, $B = \sigma_r$, $C = \sigma_{r'}$.

	The lemma follows by combining Equations \eqref{eq:cor-computed-locally-1} and \eqref{eq:cor-computed-locally-2} with the observation that $r = \sum_i r|_i$. 
\end{proof}

\begin{lemma}[Commutators]\label{lemma:commutators}
	Suppose $B$ commutes with $[A,C]$. Then $[A,B][A,C] = [A,BC]$. 
\end{lemma}

\begin{proof}
	Write $[A,BC]$ as $A(BC)A\1(BC)\1$. 
	Note that by definition, $AB = [A,B]BA$. Then we have
	\begin{align*}
		[A,BC]
		&= A(BC)A\1(BC)\1
		\\&= ABCA\1C\1B\1
		\\&= [A,B]BACA\1C\1B\1
		\\&= [A,B]B[A,C]B\1
		\\&= [A,B][A,C]\;,
	\end{align*}
	where the last line follows from commutation of $B$ and $[A,C]$. 
\end{proof}

\subsection{Rotated and stretched stabilizer games}
\label{sec:rotated-game}
In this section we define stretched stabilizer games which formalize the notion of distributing a stabilizer game out over long ``paths''.  One property of stretched games is that players on far ends of the paths have outputs which require correction according to a function of the outputs along the intermediate points in the paths. We introduce a notion of \emph{rotated stabilizer game} that captures this scenario by allowing the players to report an additional ``rotation string''.

% Suppose $\ell$ players are about to play a stabilizer game, using the honest strategy.  Suddenly Nature intervenes in the following contrived way: She applies a many-qudit Pauli to the player's shared state, tells the verifier which Pauli She applied, but doesn't tell the players. How can the verifier run the game? 

% wish to run a stabilizer game 
% in a scenario where pauli errors may be applied everywhere

% We introduce a notion of \emph{rotated stabilizer game} that captures this scenario by allowing the players to report an additional ``rotation string'' (Nature's choice of Pauli). 
% \tnote{I admit I get a little confused by this description. Can it be improved?}

\begin{definition}[Rotated stabilizer game]
	Given a stabilizer game $\game = (X_j,\{w_x,b_x\})$  
	the rotated stabilizer game associated with $\game$, $\game^R$, is defined as follows. 
	For each $j\in\{1,\ldots,\ell\}$ and question $x_j \in X_j^{m_j}$, the $j$-th player reports an answer $a_j \in \Z_d^{m_j}$ together with a rotation string $r_j \in (\Z_d^2)^{k_j}$. 
	%. For each question string $q\in \prod_i X_i$, we assign a Pauli operator $\s(Q) = \otimes_i q_i$ (recall that the $X_i$ are sets of Pauli operators). Given the rotation string $r$, we assign to eaach question a 
The win condition \eqref{eq:stabilizer-win-condition} is replaced by the condition
	\begin{equation}\label{eq:rotated-stabilizer-win-condition}
		 w_x \cdot (a - \cor_r(x)) = b_x\;,
	\end{equation}
where $r=(r_1,\ldots,r_\ell)$.
\end{definition}

Observe that if $r$ is the $0$ vector then for any $q$, $\cor_r(q)=0$, so the win condition for the rotated game $\game^R$ reduces to the win condition for $\game$. Therefore any strategy for $\game$ implies a strategy for $\game^R$ with the same success probability. More generally, it is possible to define a strategy in $\game^R$ by having the players conjugate their observables in $\game$ by an arbitrary Pauli observable (the same for all observables), and report as rotation string the string that specifies the observable used for conjugation. 

 Using Lemma~\ref{lem:cor-computed-locally} it follows that there is a reduction in the other direction as well. Given a strategy for $\game^R$, one obtains a strategy for $\game$ by replacing the answer $(a_i,r_i)$ from the $i$-th player in $\game^R$ by the answer 
\begin{equation}\label{eq:rot-to-normal}
(a_i - \cor_{r_i}(q_i))
\end{equation}
 in $\game$. 
The following lemma summarizes this observation in terms of rigidity of the rotated game. Recall the definition of a robustly rigid game in  Definition~\ref{def:robust-rigid}.

\begin{lemma}[Rotation preserves rigidity]
\label{lem:rotation-preserves-rigidty} 
	Suppose that a stabilizer game $\game$ is robustly rigid (see Definition~\ref{def:robust-rigid}). Let $\tau= (\proj{\psi},\{ M_{x_j}\})$ be a rigid strategy and $f$ the robustness. 
	Then the rotated stabilizer game $\game^R$ is rigid in the following sense. For any strategy $\tau=(\rho',\{M_{x_j}'\})$ that has value $w'=\omega^*_\tau(\game)$ in $\game^R$, there is a strategy in $\game$ that is a coarse-graining of $(\rho',\{M_{x_j}'\})$ according to~\eqref{eq:rot-to-normal},\footnote{Here by ``coarse-graining'' we mean the strategy that is implied by requiring each player to compute the update~\eqref{eq:rot-to-normal} locally; Lemma~\ref{lem:cor-computed-locally} shows that this can always be done.} and that has value $w'$ in $\game$. In particular, up to local isometries the state $\rho'$ is within distance $f(1-w')$ of $\proj{\psi}$.
	\end{lemma}

We introduce a notion of ``stretched'' rotated game, that will be useful when we relate circuit games to stablizer games. 

\begin{definition}[Stretched stabilizer game]\label{def:stretched-game}
	Let $\game = (X_j,\{w_x,b_x\})$ be a stabilizer game, and $\Gamma = (\Gamma_1,\ldots,\Gamma_\ell)$ an $\ell$-tuple of finite sets, such that for $j\in\{1,\ldots,\ell\}$, $\Gamma_j$ has $k_j$ designated elements $(u_{j,1},\ldots,u_{j,k_j})$. Each element of $\Gamma_j$ is used to index one out of $|\Gamma_j|$ qudits that are supposed to be held by the $j$-th player. To $\game$ and $\Gamma$ we associate a ``stretched'' game $\game_\Gamma^S$ as follows. In $\game_\Gamma^S$ the parameter $k_j$ is replaced by $k'_j=|\Gamma_j|$. For any $k_j$-qudit Pauli observable asked to player $j$ in $\game$, there is a $k'_j$-qubit Pauli observable in $\game_\Gamma^S$ such that the observable acts as the identity on the additional $(k'_j-k_j)$ qubits.  The win condition in $\game_\Gamma^S$ is the same as in $\game$. 
	\end{definition}

Given a stabilizer game $\game$ and sets $\Gamma=(\Gamma_1,\ldots,\Gamma_\ell)$, we write $\game_\Gamma^{S,R} = (\game_{\Gamma}^S)^R$ for the rotated stretched stabilizer game associated with $\game$ and $\Gamma$.

\subsection{Repeated games}
\label{sec:repeated-game}

For an integer $r\geq 1$ we consider the game that is obtained by repeating a stabilizer game $r$ times in parallel, with $r$ independent sets of $\ell$ players (that may share a joint entangled state).

\begin{definition}\label{def:n-rotated}
Let $\game$ be an $(\ell,k,m)$  stabilizer game, and $r\geq 1$ an integer. The \emph{$r$-fold repetition of $\game$} is the $(r\ell,k)$ stabilizer game $\game_r$ that is obtained by executing $\game$ independently in parallel with $r$ groups of $\ell$ players. More formally, the input distribution $\pi_r$ in $\game_r$ is the direct product of $r$ copies of the input distribution $\pi$ in $\game$, and the win condition in $\game_r$ is the AND of the win conditions in each copy of $\game$. 
\end{definition}

For purposes of randomness expansion, in Section~\ref{sec:randomness-generation} we consider repeated games for which the input distribution $\pi_r$ is not exactly the direct product of $r$ copies of $\pi$, but a derandomized version of it. Similarly, to achieve better robustness, instead of the AND of the winning conditions we may consider a win condition that is satisfied as soon as sufficiently many of the win conditions for the subgames are satisfied. These modifications are explained in Section~\ref{sec:randomness-from-stabilizer}. 

\subsection{The Magic Square game}

For concreteness we give two examples of stabilizer games, the Memin-Peres Magic Square game~\cite{mermin1990simple} and the Mermin GHZ game~\cite{mermin1990extreme}. The former is given for illustration; the latter will be used towards randomness expansion in Section~\ref{sec:randomness-generation}.

\begin{definition}[Magic Square game]
\label{def:magic-square}
Consider the following $3\times 3$ matrix, where each entry is labeled by a two-qubit Pauli observable:
\begin{equation}
  \label{eq:ms}
  \begin{bmatrix}
    xi & ix & xx \\
    iz & zi & zz \\
    xz & zx & yy
  \end{bmatrix}\;.
\end{equation}
The Magic Square game is a $(2,2,2)$ stabilizer game over $2$-dimensional qubits defined as follows. The sets $X_1 = X_2$ each contain $6$ pairs of two qubit-Pauli observables, the first two pairs indicated in each of the rows and columns of~\eqref{eq:ms}. The distribution $\pi$ is uniform on pairs of entries associated with the same row or column. For any query $x=(x_1,x_2)$ each player reports two bits associated with the two observables it was asked about. We can associate a third bit to the third observable in the corresponding row or column by taking the parity of the first two bits, except for the case of the third column, where we take the parity plus $1$. The constraint $w_x\cdot a = b$ expresses the constraint that, whenever the questions $x_1$, $x_2$ are associated with a row and column that intersect in an entry of the square, the outcomes associated with the intersection should match. 
\end{definition}

\begin{definition}[Honest strategy in the Magic Square game]
In the honest strategy, the two players share two EPR pairs. Upon reception of a question that indicates two commuting two-qubit Pauli observables, the player measures both observables on her qubits and reports the two outcomes. 
\end{definition}

The following robustness result is shown in~\cite{wu2016device}.

\begin{theorem}\label{thm:ms-robust}
The Magic Square game is robustly rigid, with respect to the honest strategy and with robustness $f(\delta)=O(\sqrt{\delta})$. 
\end{theorem}

Next we recall the Mermin GHZ game.

\begin{definition}[GHZ game]\label{def:ghz}
The game $\GHZ$ is a $(3,1,1)$ stabilizer game over $2$-dimensional qubits defined as follows. The sets $X_1 = X_2 = X_3 = \{0,1\}$. The distribution $\pi$ is uniform over the set $\{(0,0,0),(0,1,1),(0,1,1),(1,0,1)\}$. For all queries $x$ the vector $w_x = (1,1,1)$. For $x=(0,0,0)$, $b_x=0$, and for all other $x$, $b_x=1$. 
\end{definition}

It is well-known that there is a honest strategy based on making Pauli measurements on a GHZ state $\ket{\psi_\GHZ} = \frac{1}{\sqrt{2}}(\ket{000}+\ket{111})$ (which can be prepared in depth $3$) and that succeeds with probability $1$ in the game $\GHZ$.

%-------------------------%
\section{Lightcone arguments for low-depth circuits}
\label{sec:lightcone-circuit}

%!TEX root = main.tex

Let $N\geq 1$ be an integer. We write $\grid_N$ for the set $\{1,\ldots,N\}^2$, that we often identify with the ``grid graph'' of degree $4$, which is the graph on this vertex set with an edge between $(i,j)$ and $(i\pm 1,j\pm 1)$, with addition taken modulo $N$. (As a matter of notation we often identify a graph with its vertex set.)
 
For an integer $0 \leq L \leq N$  and $u\in \grid_N$ we write $\bx_L(u)$, or $\bx(u)$ when $L$ is implicit, for the set $\bx_L(u) = \{u\} + \{-L,\ldots,L\}^2\subseteq\grid_N$ (with addition again taken $\bmod N$). In other words, $\bx_L(u)$ is the closed ball of radius $L$ around $u$ in the $L_\infty$ metric.

\subsection{Lightcones}
\label{sec:lightcones}

Recall that the circuits that we consider are defined over an arbitrary gate set with bounded fan-in $K$.
Given a circuit $\m C$, we introduce the natural notion of a \emph{circuit graph}, with vertices at the gates and edges along the wires. 

\begin{definition}
Let $\m C$ be a circuit. The \emph{circuit graph} associated with $\m C$ is a directed graph on vertex set $V = \mathcal{I}\cup \mathcal{U} \cup \mathcal{O}$. Here $\mathcal{I}$ contains one vertex for each input wire, $\mathcal{O}$ contains one vertex for each output wire, and $\mathcal{U}$ contains one vertex for each gate. There is an edge from $u$ to $v$ if the output of $u$ is an input of $v$. In particular, all vertices of $\mathcal{I}$ are sources (have indegree $0$) and all vertices of $\mathcal{O}$ are sinks (have outdegree $0$). We call vertices in $\mathcal{I}$ \emph{input vertices} and vertices in $\mathcal{O}$ \emph{output vertices}. 
\end{definition}

We typically consider circuits that are spatially local on a 2D grid, in which case we identify the input and output sets of the graph with a grid, i.e.\ $\mathcal{I}=\mathcal{O}=\grid_N$ for some integer $N\geq 1$. Note that the circuit graph of a circuit with fan-in $K$ has in-degree bounded by $K$, but has no a priori bound on the out-degree.  

\begin{definition}
Let ${\m C}$ be a circuit. For a vertex $v$ in the circuit graph define its \emph{backward lightcone} $L_b(v)$ as the set of input vertices $u$ for which there exists a path in the circuit graph from $u$ to $v$. For an input vertex $u$ define the \emph{forward lightcone of $u$}, $L_f(u)$, as the set of output vertices $v$ such that $u\in L_b(v)$.
\end{definition}

The following lemma is established in Section 4.2 of \cite{bravyi2017quantum} during the proof of their Theorem 2. We include the short proof for completeness.

\begin{lemma}[\cite{bravyi2017quantum}]\label{lem:lightcones}
 Let $\m C$ be  a circuit that has depth $D$ and maximum fan-in $K$. Then the following hold:
\begin{itemize}
\item 	All backward lightcones are small. That is, for every vertex $v$ of the circuit graph, $\abs{L_b(v)} \leq K^D$.
\item 	Most forward lightcones are small. That is, for any $\mu\in(0,1)$,
	\begin{equation}
		\Pr_{v}[L_f(v) \geq \mu\1K^D] \leq \mu\;,
	\end{equation}
	where the probability is taken over the choice of a uniformly random input vertex $v\in \mathcal{I}$.
\end{itemize}
\end{lemma}

\begin{proof}
Every path in the circuit graph has length at most $D$. Each vertex has indegree at most $K$. Then for any fixed vertex $v$, there are at most $K^D$ distinct paths through the circuit graph ending at $v$. Therefore, $\abs{L_b(v)} \leq K^D$.

Now consider the directed graph with an edge from $u$ to $v$ if $u \in L_b(v)$. The in-degree of vertex $v$ is equal to $\abs{L_b(v)}$ while its out-degree is $\abs{L_f(v)}$. Each vertex has an in-degree of at most $K^D$, so there are at most $nK^D$ edges in the graph, where $n$ is the number of output wires for the circuit. Fix $\mu\in (0,1)$. By Markov's inequality, at most $\mu n$ vertices may have out-degree at least $\frac1\mu K^D$.
\end{proof}

\subsection{Input patterns}
\label{sec:patterns}

We introduce a method to ``plant'' queries to the players in a stabilizer game into the input to a circuit. The main definition we need is of an \emph{input pattern}, that specifies locations for each players' question, as well as paths between these locations. These paths, or ``stars'', will be useful in the design of a quantum circuit that implements the players' strategy as a low-depth quantum circuit; this is explained in Section~\ref{sec:game-completeness}.

\begin{figure}[htb!]
\centering%
\hfill
\includegraphics[scale=0.8, angle = 0]{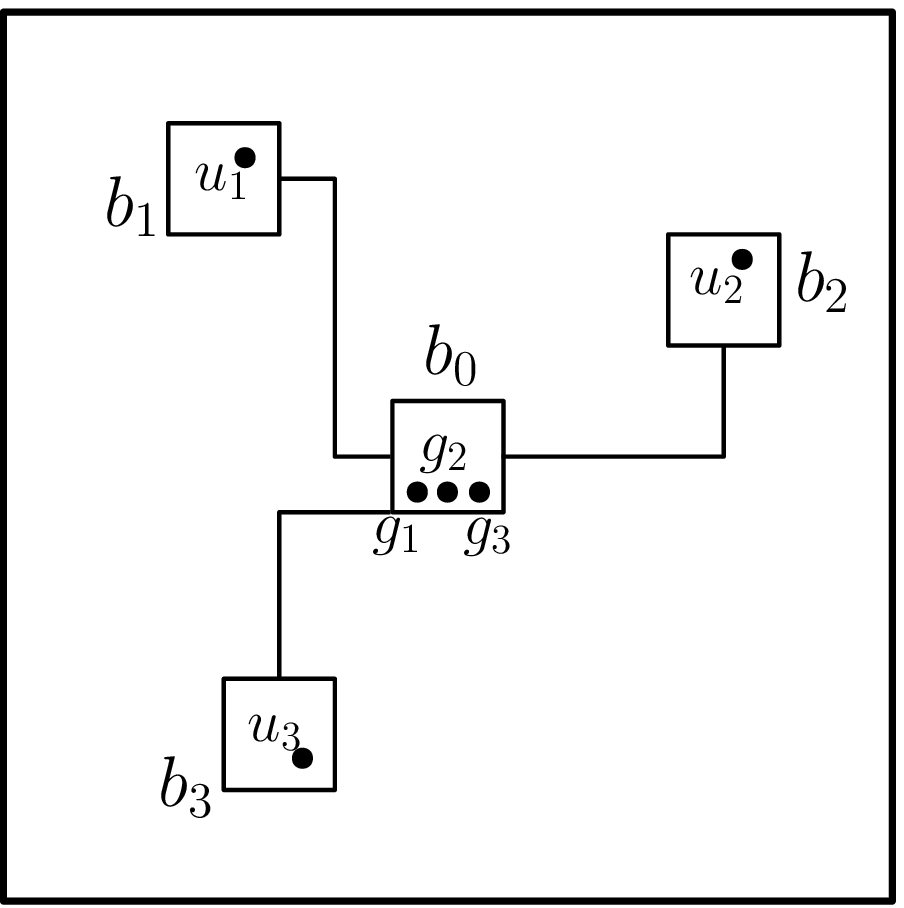}
\hfill
\includegraphics[scale=0.8, angle = 0]{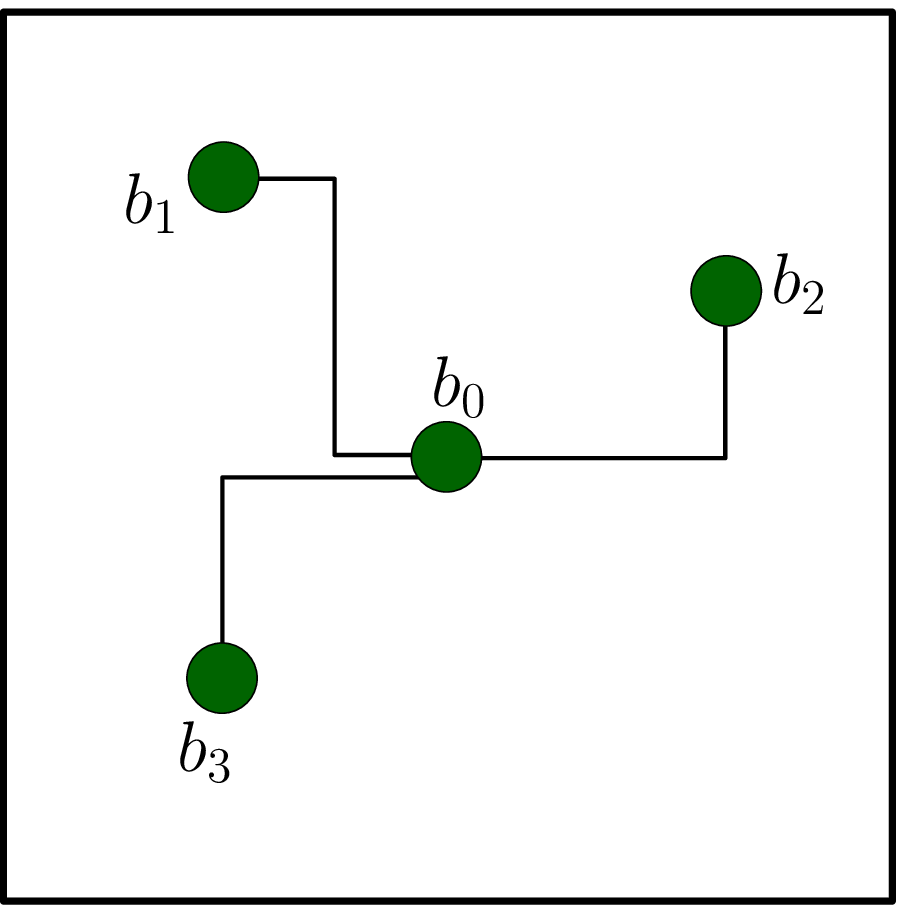}
\hfill
\caption{A star centered at box $b_0$, connecting $\ell=3$ boxes $b_1$, $b_2$, $b_3$ in $\grid_N$ (grid edges are not shown on the picture). The paths may be extended inside each box to connect the vertices $u_1,u_2,u_3$ to $g_1$, $g_2$, $g_3$ in an arbitrary way. Such connections will be used in Section \ref{sec:game-completeness} to define low-depth measurements which distribute a three-qubit state at sites $g_1,g_2,g_3$ among the qubits $u_1,u_2,u_3$.
On the right, we show the contraction of the star to a star graph. The paths are contracted to single edges (shown by thick lines) and the boxes are contracted to single vertices (shown by filled-in circles).}
\label{fig:star}
\end{figure}

\begin{definition}[Star]
	See Figure \ref{fig:star}.
	We say that a subset of $\grid_N$ is a \emph{box} if it is equal to $\bx_L(u)$ for some integer $L$ and vertex $u$. A \emph{star} $\Gamma$ is a collection of disjoint boxes together with a collection of disjoint paths such that
	\begin{itemize}
	\item each path has its endpoints on the boundaries of boxes, and
	\item contracting each box to a single vertex, and each path to a single edge, results in a a star graph, i.e. a graph that has $\ell$ vertices of degree one, one vertex of degree $\ell$, and no other vertices.
	\end{itemize}
We use the term \emph{central box} to refer to the unique box which contains one endpoint of every path. If $b_0$ is the central box, we may say that the star $\Gamma$ is \emph{centered at $b_0$}. By abuse of notation, we often write $\G$ to refer to the set of vertices contained in the paths and boxes of $\G$.
\end{definition}

The following definition captures exactly the amount of information that we need to remember about a given circuit $\m C$ in order to talk about the spread of correlations within $\m C$ --- we will forget everything about the circuit except some information about its lightcones. 

\begin{definition}[Input pattern]\label{def:pattern}
Let $\game$ be an $(\ell,k,m)$ stabilizer game. Let $1\leq r \leq N$ be integer. 
An \emph{input pattern} associated with $(\game,N,r)$ is a tuple $\pattern = \{(u^{(i)},\Gamma^{(i)})\}_{1\leq i \leq r}$ such that 
\begin{itemize}
	\item each $u^{(i)}= (u^{(i)}_1,\ldots,u^{(i)}_\ell)$ is an $\ell$-tuple of vertices of $\grid_N$, which we refer to as \emph{input locations},
	\item each $\G^{(i)}$ is a star,
	\item the vertices of $u^{(i)}$ are contained in distinct noncentral boxes of $\G^{(i)}$. For a vertex $u$, we write $\bx(u)$ for the box that contains $u$. 
\end{itemize}
\end{definition}

\begin{definition}[Circuit specification]
A \emph{circuit specification} $\spec$ on $\grid_N$ is a triple $\spec=(L_f,\bad_{in},\bad_{out})$ such that for all $u\in \grid_N$,  $L_f(u) \subseteq \grid_N$ is a set called the \emph{forward lightcone} associated with $u$, and $\bad_{in},\bad_{out} \subseteq \grid_N$ are sets called the \emph{bad input set} and \emph{bad output set} respectively. 
\end{definition}

\begin{definition}
For integer $B,R_{in},R_{out}$ we say that a circuit specification $\spec = (L_f,\bad_{in},\bad_{out})$ on $\grid_N$ is \emph{$(B,R_{in},R_{out})$-bounded} if the following hold: $|\bad_{in}|\leq R_{in}$, $|\bad_{out}|\leq R_{out}$, and for every $u\in \grid_N\backslash \bad_{in}$ it holds that $|L_f(u)|\leq B$. 
\end{definition}

% The intuition for the sets $\bad_{in}$ and $\bad_{out}$ is that the first set, $\bad_{in}$, designates vertices that are ``bad'' due to the circuit --- they have unusually large lightcones --- whereas the second, $\bad_{out}$, will be used to specify vertices in the grid that are ``bad'' due to other choices, e.g. they are in the lightcone of a different input location; see the proof of Lemma~\ref{lem:good-rgeneral}. 

Given a fixed circuit specification, the following definition captures the conditions that are required for an input pattern so that the circuit game associated with that input pattern can be reduced to a  nonlocal game (the reduction is explained in Section~\ref{sec:circuit-games}).

The intuition to keep in mind for the definition is as follows:
each player in the nonlocal game receives her input from one of the input locations and puts her output along the paths of the star. Each player also puts some outputs inside their box of the star. In order for it to be possible to implement the strategy locally, we must have the outputs of each player be causally independent of the inputs of the other players. We ensure this by checking that the forward lightcone of one player's input misses the locations of the other players' outputs.

\begin{definition}[Causality-respecting patterns]\label{def:good-pattern}
Let $\spec = (L_f,\bad_{in},\bad_{out})$ be a circuit specification.
Let $\pattern=\{(u^{(i)}, \G^{(i)})_i\}$ be an input pattern. We say that a pair $(u^{(i)},\G^{(i)})$ is \emph{individually-\spec-causal} with respect to $\pattern$ if the following hold:\footnote{Recall that we identify a star $\G$ with the union of the vertex sets of its paths and boxes.}
\begin{enumerate}[label=(\alph*)]
	\item 
	\label{item:miss-own-lightcone}
	For each $k$, the forward lightcone of $u_k^{(i)}$ misses $\G^{(i)}$, except possibly near $u_k^{(i)}$. More precisely,
	 $L_f(u_k^{(i)}) \cap \G^{(i)} \seq \bx(u_k^{(i)})$.
	\item 
	\label{item:miss-other-lightcone}
	For all $(u^{(j)},\G^{(j)})\in \pattern$ (with $j \neq i$) and for all $k$, the forward lightcone of $u^{(j)}_k$ misses $\G^{(i)}$ entirely, i.e.\ $L_f(u^{(j)}_k)\cap \G^{(i)}=\emptyset$.
	\item \label{item:miss-bad-out}
	% \tnote{every star in what? Should this be $\Gamma^{(j)}$? What about $\Gamma$?} 
	$\G^{(i)}$ misses $\bad_{out}$, i.e. $\G^{(i)}\cap \bad_{out} = \emptyset$. 
\end{enumerate}
\noindent
Furthermore, we say that a pair $(u,\Gamma)$ is \emph{$\spec$-valid} if the following conditions hold. 
\begin{enumerate}[label=(\alph*),resume]
	\item \label{item:miss-bad-in}
	Every input location lies outside of $\bad_{in}$, i.e. $u_k^{(i)}\cap \bad_{in} = \emptyset$ for all $k,i$.
\end{enumerate}
\noindent
We say that an input pattern $\pattern$ is \emph{$\spec$-causal} if every $(u^{(i)},\G^{(i)})\in \pattern$ is individually-\spec-causal and \spec-valid with respect to \pattern..

% Say that a pair $(u,\Gamma)\in\pattern$ is \emph{valid} if for all $j\in\{1,\ldots,\ell\}$, $u_j\notin\bad_{in}$. Then the following conditions should hold for all valid $(u,\Gamma)$ and $(u',\Gamma')$ in $\pattern$:
% \begin{itemize}
% \item For every $j\in\{1,\ldots,\ell\}$, $L_f(u_j) \cap \Gamma \subseteq \bx(u_j)$;
% \item If $(u,\Gamma)\neq (u',\Gamma')$, then for every $j\in\{1,\ldots,\ell\}$, $L_f(u_j) \cap \Gamma' = \emptyset$;
% \item For every $j,k\in\{1,\ldots,\ell\}$ such that $u_j\neq u'_k$, $L_f(u_j)\cap \bx(u'_k)=\emptyset$ and $\bx(u_j)\cap \bx(u'_k)=\emptyset$.
% \item $(\cup_j \bx(u_j) \cup \Gamma) \cap \bad_{out} = \emptyset$. 
% \end{itemize}
% If a pattern is not good, we say that it is \emph{bad}.
\end{definition}

% Note that in Definition~\ref{def:good-pattern}, the parameter $L$ that determines the box size is left implicit in the notation. 

Finally, we introduce a distribution on input patterns so that for any  circuit specification \spec\ that is $(B,R_{in},R_{out})$-bounded for sufficiently small parameters $B$, $R_{in}$, and $R_{out}$, a sample from the distribution  gives an \spec-causal pattern with high probability (see Section~\ref{sec:single-input-patterns} and Section~\ref{sec:arbitrary-input-patterns}).

\begin{definition}[Random input patterns]\label{def:sample-dist}
Let $L,N\geq 1$, $\ell\geq 1$, and $1\leq r \leq N$ be integer such that $3L\sqrt{\ell+1} \leq M=\lfloor N/\sqrt{r}\rfloor$. Divide $\grid_N$ in $r$ disjoint squares $S^{(1)},\ldots,S^{(r)}$ of side length $M$ each.\footnote{It does not matter where these squares are located, as long as they do not overlap.} Partition each square into $T = \lfloor \frac{M}{2L+1} \rfloor^2$ boxes of side length $(2L+1)$, in an arbitrary way. For each possible choice of $(\ell+1)$ distinct boxes $b_0,b_1,\ldots,b_\ell$ within a square, fix a collection $\textsc{stars}(b_0,\ldots,b_\ell)$ of $L/\ell$ stars such that
\begin{itemize}
	\item  each star has $b_0$ as its central box and $b_1,\ldots,b_\ell$ as its other boxes,
	\item the total length of the paths in any star is at most $2\ell M$, and
	\item the paths of the distinct stars are vertex-disjoint.
\end{itemize}

Consider the following distribution $\mathcal{D}^{(r)}(N,L)$ on input patterns on $\grid_N$. 
For each $i\in \{1,\ldots,r\}$ select $b^{(i)}_0,\ldots,b^{(i)}_\ell$ uniformly at random among the $T$ boxes that partition the $i$-th square $S^{(i)}$, for $j\in\{1,\ldots,\ell\}$, a vertex $u^{(i)}_j$ uniformly at random within the $j$-th selected box.
Finally, select a star $\Gamma^{(i)} \in \textsc{stars}(b^{(i)}_0,b^{(i)}_1,\ldots,b^{(i)}_\ell)$ uniformly at random.
Return the input pattern $\pattern = \{(u^{(i)},\Gamma^{(i)})\}_{1\leq i \leq r}$. 
\end{definition}

\subsection{Single-input patterns}
\label{sec:single-input-patterns}
We'd like to show that patterns in the support of $\m D^{(r)}$ are ``very nearly'' \spec-causal for most \spec\ in the sense that removing only an exponentially small fraction of inputs yields an \spec-causal pattern. To warm up, we argue that for any $(B,R_{in},R_{out})$-bounded circuit specification $\spec$, an input pattern sampled from the distribution $\mathcal{D}^{(1)}$ introduced in Definition~\ref{def:sample-dist} is \spec-causal with high probability. 
We  use this single-input analysis later to show that in a many-input pattern, most of the inputs are individually-\spec-causal. 

In this subsection only, we use $M$ instead of $N$ to denote the grid size. We do this because the distribution $\m D^{(r)}(N,L)$ can be (informally) thought of as the direct product of $r$ copies of $\m D^{(1)}(M,L)$, and the former is of greater interest to us. 

%\begin{definition}
%Let $M\geq 1$ and $1\leq L\leq M/4$ be integer. Partition $\grid_M$ into $T = \lfloor \frac{M}{2L+1} \rfloor^2$ boxes of side length $(2L+1)$, in an arbitrary way. For any pair of boxes, fix $L$ node-disjoint paths from one box to the other, in an arbitrary way. Define a distribution $D$ on pairs $(u,\Gamma)$ of a pair of vertices $u=(u_1,u_2)\in \grid_M$ and a path $\Gamma$ from $u_1$ to $u_2$ as follows. Choose two boxes uniformly at random, and let $u_1$, $u_2$ be their centers. Choose $\Gamma$ uniformly at random from the $L$ designated paths. 
%\end{definition}

\begin{lemma}\label{lem:good-r1}
Let $M\geq 1$, $1\leq B,L\leq M/4$ and $0\leq R_{in},R_{out} \leq M^2$ be integer. Let $\spec= (L_f,\bad_{in},\bad_{out})$ be a circuit specification for $\grid_M$ that is $(B,R_{in},R_{out})$-bounded. Let $\pattern = \set{(u,\Gamma)}$ be drawn from the distribution $\m D^{(1)}$ introduced in Definition \ref{def:sample-dist}. Then the probability that $(u,\Gamma)$ is not individually-$\spec$-causal with respect to \pattern\ is $O(L^2(R_{out}+B)/M^2+(R_{out}+B)/L)$.
Moreover, the probability that \pattern\ is not $\spec$-valid is $O(R_{in}/M^2)$.  
Overall, the probability that \pattern\ is not \spec-causal is at most
\begin{equation}
O\left(L^2(R_{out}+B)/M^2 + (R_{out} + B)/L + R_{in}/M^2\right),
\end{equation}
where the $O$ notation hides  factors polynomial in $\ell$. 
\end{lemma}

\begin{proof}
We check all conditions in Definition~\ref{def:good-pattern}.
Since $\pattern$ contains only one (input, star) pair, condition \ref{item:miss-other-lightcone} (which restricts the interactions between pairs) is satisfied automatically.
 
Now we check conditions \ref{item:miss-own-lightcone} and \ref{item:miss-bad-out}. Call a box bad if it intersects $\bad_{out}$. Under $\mathcal{D}^{(1)}$ there are $\lfloor M/(2L+1) \rfloor^2 \geq 1/4 (M/2L)^2$ possible box locations. By a union bound, the probability that any box is bad is at most  $16L^2R_{out}/M^2 =  O(L^2R_{out}/M^2)$.
There are at least $L/\ell $ possible choices for the paths of $\Gamma$. Since all such choices are disjoint, again by a union bound the probability that $\Gamma \cap X \neq \emptyset$ for some subset $X$ is at most $\ell\abs X / L$. 
% \ell R_{out}/L = O(R_{out}/L)$.
Letting $X = \bad_{out} \cup \bigcup_i L_f(u_i)$, so that $\abs X \leq R_{out} + \ell B$, we see that the probability of violating condition \ref{item:miss-own-lightcone} or condition \ref{item:miss-bad-out} is at most $O\left(\frac{\ell R_{out} + \ell^2 B}{L}\right)$.
 
Similarly, since the $u_j$ are chosen independently, for any $u\neq u'\in\{u_0,u_1,\ldots,u_\ell\}$ the probability that $L_f(u)\cap \bx(u') \neq \emptyset$ is $16L^2B/M^2 =  O(L^2B/M^2)$.

Finally we check condition \ref{item:miss-bad-in}. Any $u_j$ is chosen independently among $ (2L+1)^2 \geq M^2/8 = \Omega(M^2)$ possibilities, so the probability that $u_j\in \bad_{in}$ is at most $8R_{in}/M^2 =  O(R_{in}/M^2)$; we conclude by the union bound, and absorb the parameter $\ell$ in the $O(\cdot)$. 
\end{proof}

\subsection{Arbitrary-input patterns}
\label{sec:arbitrary-input-patterns}

We extend the argument from the previous section to the case where the input pattern contains more than one input. 

\begin{lemma}[Random input patterns are usually causal]\label{lem:good-rgeneral}
Let $N\geq 1$, $1\leq r \leq N$ and $1\leq B,R_{in},L\leq N/4$ be integer. Let $\spec= (L_f,\bad_{in},\emptyset)$ be a circuit specification for $\grid_N$ that is $(B,R_{in},0)$-bounded. Then the probability that an input pattern $\pattern$ chosen according to $\mathcal{D}^{(r)}(N,L)$ (as defined in Definition \ref{def:sample-dist}) is not $\spec$-causal is at most $O(r^2B(r(L^2+R_{in})/N^2 + 1/L ))$.

\end{lemma}

%Lemma \ref{lem:good-rgeneral} will be our tool of choice for producing good input patterns for the purposes of certifiable randomness generation via circuit games.

\begin{proof}
Let $\pattern = \{(u^{(i)},\Gamma^{(i)})\}$ be an input pattern chosen according to $\mathcal{D}^{(r)}(M,L)$. For $i \in\{ 1,\ldots,r\}$ we let $\pattern^{(i)}$ be the single-pair pattern $\{(u^{(i)},\Gamma^{(i)})\}$. 
Let $X_i$ be the indicator variable that the pair $(u^{(i)},\Gamma^{(i)})$ is not individually-$\spec$-causal with respect to \pattern. Let $Y_i$ be the indicator that $(u^{(i)},\Gamma^{(i)})$ is not  $\spec$-valid. To see that \pattern\ is \spec-causal, it suffices to check that each input is \spec-valid and individually-\spec-causal with respect to \pattern. This is true if and only if $\sum_i X_i=0$ and $\sum_i Y_i =0$. We first bound the latter event. 

\begin{claim}
It holds that
\begin{equation}\label{eq:xdepy-2}
\Pr\Big(\sum_j Y_j \neq 0  \Big)   \,=\, O\big(r^2R_{in}/N^2\big)\;.
\end{equation}
\end{claim}

\begin{proof}
Applying the second bound in Lemma~\ref{lem:good-r1} and using that $M=\lfloor N/\sqrt{r}\rfloor$ it follows that for any $i\in\{1,\ldots,r\}$,
\begin{equation*}
\Prob{Y_i \neq 0} 
\,=\,\Pr\big(\pattern^{(i)} \text{ is not } \spec\text{-valid}\big)
\,\leq\, 8R_{in}/M^2  
\,=\, O\big(rR_{in}/N^2\big)\;.
\end{equation*}
The claim follows by a union bound over the $r$ patterns $\pattern^{(i)}$.
\end{proof}

Next we turn to the $X_i$. 

\begin{claim}
\begin{equation}
\Pr\Big(\sum_i X_i \neq 0\Big)\,=\, O\big(r^2B(rL^2/N^2 + 1/L)\big) + \sum_i \Pr\Big(\sum_{j \neq i} Y_j \neq 0\Big)\;.\label{eq:xdepy}
\end{equation}
\end{claim}

\begin{proof}
For any $i\in\{1,\ldots,r\}$ let  
$$\bad_{out}^{(i)} = \bigcup_{k\neq i} \Big(\cup_j L_f(u_j^{(k)})\Big)\;,$$
and define a specification $\spec^{(i)}= (L_f,\bad_{in},\bad_{out}^{(i)})$. With these definitions, it follows that
\begin{align}
	\Pr(X_i = 0) & = \Pr \big(\pattern^{(i)} \text{ is } \text{individually-}\spec\text{-causal}\big) \nonumber \\
	&\geq  \Pr \big(\pattern^{(i)} \text{ is } \text{individually-}\spec^{(i)}\text{-causal}\big) \;.\label{eq:probineq}
\end{align}
Indeed condition \ref{item:miss-bad-out} of being  individually-$\spec^{(i)}$-causal (see Definition~\ref{def:good-pattern}) implies all conditions of being individually-$\spec$-causal for  $\spec= (L_f,\bad_{in},\emptyset)$. 

In the event that $\pattern^{(j)}$ is $\spec$-valid for all $j\neq i$ (that is, when $\sum_{j \neq i} Y_j = 0$) we know that $L_f(u_j^{(k)}) \leq B$ for each $(j,k) \in \{1,\ldots,\ell\}\times \{1,\ldots,r\}$, and thus that  $|\bad_{out}^{(i)}| \leq \ell rB = O(rB)$. Using that the marginal distribution of a single pair $(u^{(i)},\Gamma^{(i)})$ from $\pattern$ is equal to $\mathcal{D}^{(1)}(M,L)$ (when seen as a distribution on the square $S^{(i)}$ associated with $(u^{(i)},\Gamma^{(i)})$), it follows from the bound in Lemma~\ref{lem:good-r1} that 
for any $i\in\{1,\ldots,r\}$,
\begin{equation}
\Pr \Big(X_i \neq 0  \Big| \sum_{j \neq i} Y_j = 0  \Big)  \,=\, O\big(rB(rL^2/N^2 + 1/L)\big)\;.\label{eq:xdepy-1}
\end{equation}

Applying the union bound, 
\begin{align}
\Pr\Big(\sum_i X_i \neq 0\Big)
&\leq   \sum_i\, \Prob{X_i \neq 0}\notag\\
&\leq \sum_i \Pr\Big(X_i \neq 0\Big| \sum_{j \neq i} Y_j = 0\Big) + \Pr\Big(\sum_{j \neq i} Y_j \neq 0\Big)\notag\\
&\leq O\big(r^2B(rL^2/N^2 + 1/L)\big) + \sum_i \Pr\Big(\sum_{j \neq i} Y_j \neq 0\Big)\;,\notag
\end{align}
where the last line follows from~\eqref{eq:xdepy-1}. 
\end{proof}

To conclude the proof of the lemma we write
\begin{align*}
	\Pr\big(\text{\pattern\ is not \spec-causal}\big) 
	&= \Pr\Big(\sum_i X_i+Y_i \neq 0\Big)\\
	&\leq \Pr\Big(\sum_i X_i \neq 0\Big)  + \Pr\Big(\sum_j Y_j \neq 0  \Big) \\
	& \leq O\big(r^2B(rL^2/N^2 + 1/L)\big) + \sum_i \Pr\Big(\sum_{j \neq i} Y_j \neq 0\Big) + O\big(r^2R_{in}/N^2\big) \\
	&\leq  O\big(r^2B(rL^2/N^2 + 1/L)\big) + O\big(r^3R_{in}/N^2\big) + O\big(r^2R_{in}/N^2\big) \\
	&= O\left(r^2B(r(L^2+R_{in})/N^2 + 1/L )\right)\;,
\end{align*}
where the third line uses~\eqref{eq:xdepy} and~\eqref{eq:xdepy-2} and the fourth uses~\eqref{eq:xdepy-2}. 
\end{proof}

The previous lemma shows that a random input pattern \pattern\ is \spec-causal with high probability. In this case we can define a game from \pattern\ so that in the game, a shallow circuit with specification \spec\ can be simulated by a set of spacelike-separated players. This simulation is perfect when \pattern\ is exactly \spec-causal. More generally, a weaker simulation argument still applies if a small constant fraction of inputs in \pattern\ are not \spec-causal.
The next lemma shows that this condition can be guaranteed to hold with much higher probability, exponentially close to $1$ rather than inverse-polynomially close. This bound will be used in the proof of Theorem~\ref{cor:exp-bgk-expsoundess}.
% To state the next lemma we will need a modified version of Definition \ref{def:good-pattern}.

% \tnote{Could we use a single definition? We can use the one below. We only need to append that we call a pattern good if all valid $(u,\Gamma)$ in the pattern areindividually-good. }

% \begin{definition} \label{def:individuallygood-pattern}
% 	Given an input pattern $\pattern$ we say that a pair $(u,\Gamma) \in \pattern$ is \emph{individually-good} for the circuit specification $\spec = (L_f,\bad_{in},\bad_{out})$ relative to $\pattern$ if the following conditions hold. Given an input pattern $\pattern$ we say that a pair $(u',\Gamma') \in \pattern$ is \emph{valid} if for all $j\in\{1,\ldots,\ell\}$, $u'_j\notin\bad_{in}$. Then for all valid $(u',\Gamma') \in \pattern$ with $(u,\Gamma) \neq (u',\Gamma')$  the following conditions must hold:
% 	\begin{itemize}
% 		\item For every $j\in\{1,\ldots,\ell\}$, $L_f(u_j) \cap \Gamma \subseteq \bx(u_j)$;
% 		\item For every $j\in\{1,\ldots,\ell\}$, $L_f(u'_j) \cap \Gamma = \emptyset$;
% 		\item For every $j,k\in\{1,\ldots,\ell\}$ such that $u_j\neq u'_k$, $L_f(u'_j)\cap \bx(u_k)=\emptyset$ and $\bx(u_j)\cap \bx(u'_k)=\emptyset$.
% 		\item $(\cup_j \bx(u_j) \cup \Gamma) \cap \bad_{out} = \emptyset$. 
% 	\end{itemize}
% 	If $(u,\Gamma)$ is not individually-good, we say that it is \emph{individually-bad}.
% \end{definition}

% Note that in Definition~\ref{def:individuallygood-pattern}, the parameter $L$ that determines the box size is left implicit in the notation.

\begin{lemma}[Random input patterns are mostly causal with high probability]\label{lem:good-rgeneral2}
	Let $N\geq 1$, $1\leq r \leq N$ and $1\leq B,R_{in},L\leq N/4$ be integer. Let $\spec= (L_f,\bad_{in},\emptyset)$ be a circuit specification for $\grid_N$ that is $(B,R_{in},0)$-bounded. Consider an input pattern $\pattern$ chosen according to $\mathcal{D}^{(r)}(N,L)$.  Let
	\begin{align*}
	\pattern_{VAL} &= \big\{(u, \Gamma) \in \pattern | (u, \Gamma) \text{ is \spec-valid} \big\}\;,\\
	\pattern_{CAUS} &= \big\{(u, \Gamma) \in \pattern | (u, \Gamma) \text{ is individually-\spec-causal with respect to $\pattern_{VAL}$} \big\}\;.
	\end{align*}
	% Define $\pattern_{VAL} = \{(u, \G) \in \pattern: (u, \G) \text{ is $\spec$-valid}\}$. 
	Then there exists universal constants $C,C'>0$ such that if $p = C' rB(rL^2/N^2 + \ell/L)$ then for any $t>0$, 
	\begin{align}
		\Pr\big(\abs{\pattern_{CAUS}} \geq r(1-p)-2t\big) &\geq 1-2\exp \paren{-t^2/8r}\;,
		\label{eq:good-rgeneral-a}
		\\
		\Pr\big(\abs{\pattern_{VAL}} \geq r (1 - C rR_{in}/N^2) - t\big) &\geq 1 - 2\exp \left ( -2t^2/r\right)\;.
		\label{eq:good-rgeneral-b}
	\end{align}
	% (and we are assuming that $p\geq 1$).
	% \noindent
	 % Furthermore, with probability at least we have that  $|\pattern_{VAL} | \geq \lfloor \rfloor - t $.
\end{lemma}

For later convenience we note that \eqref{eq:good-rgeneral-a} and \eqref{eq:good-rgeneral-b} can be combined by a union bound to obtain
\begin{equation}\label{eq:bound-on-patterncore}
	\Pr\big(\abs{\pattern_{VAL} \cap \pattern_{CAUS}} \geq r (1 - CrR_{in}/N^2 - 2p) - 3t\big) \,\geq\, 1 - 4\exp \Big(-\frac{t^2}{8r}\Big)\;.
\end{equation}

\begin{proof}
The proof relies on concentration arguments to bound the probabilities in~\eqref{eq:good-rgeneral-a} and~\eqref{eq:good-rgeneral-b}. The second bound,~\eqref{eq:good-rgeneral-b}, is easier to show, because it can be expressed as a bound on a sum of independent random variables. The following claim establishes the bound. 

\begin{claim}
There is a universal constant $C>0$ such that for any $t>0$, 
\[\Pr\big(\abs{\pattern_{VAL}} \geq r (1 - C rR_{in}/N^2) - t\big) \,\geq\, 1 - 2\exp \Big(-\frac{2t^2}{r}\Big)\;.\]
\end{claim}

\begin{proof}
For $i\in\{1,\ldots,r\}$ let $V_i$ be the indicator variable for the event that $(u^{(i)},\Gamma^{(i)})$ is $\spec$-valid. Since in $\mathcal{D}^{(r)}(N,L)$  the $(u^{(i)},\Gamma^{(i)})$ are chosen independently within the disjoint squares $S^{(i)}$, it follows from the definition of $\spec$-valid that the $V_i$ are independent.  Using the bound shown in Lemma~\ref{lem:good-r1} it follows that for any $i\in\{1,\ldots,r\}$, 
 \begin{equation}
 	\Es{}[V_i] = 1 -C R_{in}/M^2\;,
 \end{equation}
 for some constant $C>0$ and where $M=\lfloor N/\sqrt{r}\rfloor$. We conclude by applying Hoeffding's inequality:
\begin{align*}
\Pr \left (|\pattern_{VAL}|    \leq r (1 - CR_{in}/M^2 )  - t \right )&= \Pr \Big(\sum_{i=1}^r V_i    \leq r (1 - CR_{in}/M^2 )  - t \Big)\\
& \leq \Pr \Big( \frac{1}{r}\Big| \sum_{i=1}^r V_i    - \sum_{i=1}^r \Es{}[V_i]  \Big|  \geq t/r \Big) \\
&\leq 2\exp \left (-2t^2/r\right )\;.
\end{align*}
\end{proof}

The proof of the remaining bound~\eqref{eq:good-rgeneral-a} is made a little delicate by the fact that the condition that a pair $(u,\Gamma)\in\pattern$ is individually-\spec-causal is a global condition, so that $\pattern_{CAUS}$ is not directly expressible as a sum of independent random variables. To get around this, we first make a few definitions. 

Let $\pattern = \{(u^{(i)},\Gamma^{(i)})\}$ be an input pattern chosen at random according to the distribution $\mathcal{D}^{(r)}(M,L)$. 
For $i \in\{ 1,\ldots,r\}$ define
\[	\bad_{out}^{<i} \,=\, \bigcup_{\substack{k< i \text{ s.t.}\\ (u^{(k)},\Gamma^{(k)}) \text{ is \spec-valid}}} \Big(\cup_j L_f(u_j^{(k)})\Big)\;,\qquad
	\bad_{out}^{>i} \,=\, \bigcup_{\substack{k> i \text{ s.t.}\\ (u^{(k)},\Gamma^{(k)}) \text{ is \spec-valid}}} \Big(\cup_j L_f(u_j^{(k)})\Big)\;,\]
and	$\bad_{out}^{(i)} = \bad_{out}^{<i} \cup \bad_{out}^{>i}$. 	From the assumption that the circuit specification $\spec$  is $(B,R_{in},0)$-bounded, and since $\bad_{out}^{(i)}$ is defined as a union of the lightcones of only the \emph{valid} pairs $(u^{(k)},\Gamma^{(k)})$,  it follows that for all $i$, $|\bad_{out}^{(i)}|$ ,$|\bad_{out}^{<i}|$, $|\bad_{out}^{>i}|$ are each at most $rB$.

Recall that in $\mathcal{D}^{(r)}(M,L)$ each $u^{(i)}$ is chosen within a square $S^{(i)}$ of side length $M = \lfloor N/\sqrt{r}\rfloor$. We identify $S^{(i)}$ with $\grid^{(i)}_M$, and introduce a single-pair pattern  $\pattern^{(i)} = \{(u^{(i)},\Gamma^{(i)})\}$ that we think of as a pattern on $\grid^{(i)}_M$. We further define a specification 
\[\spec^{(i)}\,=\,(L_f,\bad_{in},\bad_{out}^{(i)})\]
on  $\grid^{(i)}_M$.

Let $X_i$ be the indicator variable for the event that the pair $(u^{(i)},\Gamma^{(i)})$ is not individually-$\spec^{(i)}$-causal with respect to $\pattern_{VAL}$. The following claim relates the $X_i$ to $\pattern_{CAUS}$.

\begin{claim}\label{claim:x-tail}
It holds that $|\pattern_{CAUS}| \geq \sum_{i=1}^r X_i$.
\end{claim}

\begin{proof}
Note that whenever $\pattern^{(i)}$ is individually-$\spec^{(i)}$-causal with respect to $\pattern_{VAL}$, it is also individually-$\spec$-causal with respect to $\pattern_{VAL}$. This follows by noting that for the given definition of $\bad_{out}^{(i)}$, condition \ref{item:miss-bad-out} of being individually-$\spec^{(i)}$-causal with respect to $\pattern$ implies conditions \ref{item:miss-other-lightcone} and \ref{item:miss-bad-out} of being individually-$\spec$-causal with respect to $\pattern_{VAL}$.  Therefore, $\sum_i X_i$ is an upper bound on the number of $\pattern^{(i)}$ which are not individually-$\spec$-causal with respect to $\pattern_{VAL}$. 
\end{proof}

The previous claim reduces our task to showing a high-probability lower bound on $\sum X_i$. The random variables $X_i$ are dependent. To obtain a bound, we apply a Martingale argument to two related sequences of random variables, defined as follows. First introduce specifications 
\[\spec^{(<i)}=(L_f,\bad_{in},\bad_{out}^{<i})\quad\text{and}\quad \spec^{(>r-i)}=(L_f,\bad_{in},\bad_{out}^{>(r-i)})\]
on grid $\grid^{(i)}_M$, an let $Y_i$ (resp. $Z_i$) as the indicator variable for the event that $(u^{(i)},\Gamma^{(i)})$ is not individually-$\spec^{(<i)}$-causal (resp. individually-$\spec^{(>r-i)}$-causal) with respect to $\pattern^{(i)}$. The next claim relates $Y_i$ and $Z_{r-i}$ to $X_i$. 
	
\begin{claim}\label{claim:xyz}
For each $i\in\{1,\ldots,r\}$, it holds that $X_i = Y_i \lor Z_{r-i}$.
\end{claim}

\begin{proof}
The claim follows by noting that, in Definition~\ref{def:good-pattern}, the conditions \ref{item:miss-own-lightcone} for $X_i, Y_i$ and $Z_{r-i}$ are  equivalent.  Condition \ref{item:miss-other-lightcone} for $X_i$ is equivalent to the event that condition \ref{item:miss-bad-out} for both $Y_i$ \emph{and} $Z_{r-i}$ is true. Finally, condition \ref{item:miss-bad-out} for $X_i$ is vacuous, and conditions \ref{item:miss-other-lightcone} for both $Y_i$ and $Z_{r-i}$ are vacuous. 
\end{proof}

The following claim almost finishes the proof. 

\begin{claim}\label{claim:yz}
There is a universal constant $C'>0$ such that if $p= C'rB(rL^2/N^2 + \ell/L))$ then for any $t>0$,
\begin{align*}
\Pr\Big(\sum_{i=1}^r Y_i  \geq t + rp\Big)  & \leq \exp \left ( \frac{-t^2}{2r}\right )\;,  \\
	\Pr\Big(\sum_{i=1}^r Z_i  \geq t + rp\Big) & \leq \exp \left ( \frac{-t^2}{2r}\right )\;. 
\end{align*}
\end{claim}
	
	\begin{proof}
	The proof is based on a Martingale tail bound. For any $i\in\{1,\ldots,r\}$, let $Y_{<i} = \{Y_k| k < i \}$ and $Z_{>i} = \{Z_k| k > i \}$. 
		Note that $Y_i$ (resp. $Z_i$) depends only on the underlying circuit $\m C$ together with the selection of pairs in squares $S^{(j)}$ for $j \leq i$  (resp. $j \geq i$). It follows from the bound shown in Lemma~\ref{lem:good-r1} that
	\begin{align}
	\mathbb{E}[Y_i|Y_{<i} ] 
	& = O\big( rB(rL^2/N^2 + \ell/L)\big)\;,\label{eq:y-bound-1}
		\end{align}
	and similarly
	\begin{equation}
	\mathbb{E}[Z_i|Z_{>i} ] \,=\, O\left( rB(rL^2/N^2 + \ell/L)\right)\;.\label{eq:z-bound-1}
	\end{equation}
		Let $p$ denote the maximum of the bounds on the right-hand side of~\eqref{eq:y-bound-1} and~\eqref{eq:z-bound-1}, and assume $p\leq 1$. For any $n\in\{1,\ldots,r\}$ define $\bar{Y}_n = \sum_{i=1}^n Y_i - np$ and $\bar{Z}_{r-n} = \sum_{i=r-n}^r Z_i - np$. Then
		$$\mathbb{E}[\bar{Y}_{n}|Y_{<n} ] = \mathbb{E}[Y_n - p + \bar{Y}_{n-1}|Y_{<n} ] \leq  \bar{Y}_{n-1}\;,$$
		and
		$$\mathbb{E}[\bar{Z}_{r-n}|Z_{>r-n} ] = \mathbb{E}[Z_{r-n} - p + \bar{Z}_{r-n+1}|Z_{>r-n} ] \leq  \bar{Z}_{r-(n-1)}\;.$$
		Additionally, it always holds that 
		\[|\bar{Y}_{n} - \bar{Y}_{n-1}| \leq |Y_n - p| \leq \max(1-p, p) \leq 1\;,\] 
		where the last inequality follows from the assumption that $p \leq 1$.  Similarly, 
		\[|\bar{Z}_{r-n} - \bar{Z}_{r-(n-1)}| \leq |Z_{r-n} - p| \leq \max(1-p, p) \leq 1\;.\]
		Thus, both $\bar{Y}_n$, and $\bar{Z}_{r-n}$ form super-martingale sequences for increasing $n$.  Defining $\bar{Y}_0  = \bar{Z}_r = 0$ and applying Azuma's inequality gives
			\[\Pr(\bar{Y}_{n} - \bar{Y}_0 \geq t) = \Pr(\bar{Y}_{n}  \geq t) \leq \exp \left ( \frac{-t^2}{2n \max(1-p, p)^2 }\right ) \leq \exp \left ( \frac{-t^2}{2n}\right ) \;, \]
		and 
		\[\Pr(\bar{Z}_{r-n} - \bar{Z}_r \geq t) = \Pr(\bar{Z}_{r-n}  \geq t) \leq \exp \left ( \frac{-t^2}{2n \max(1-p, p)^2 }\right ) \leq \exp \left ( \frac{-t^2}{2n}\right ) \;. \]
		Setting $n = r$ proves the claim.
		\end{proof}
	
	Using Claim~\ref{claim:xyz} it follows that  $\sum_{i=1}^r X_i \leq \sum_{i=1}^r Y_i + \sum_{i=1}^r Z_i$.  Therefore, for any $t>0$ 
		\begin{align}
		\Pr\Big(\sum_{i=1}^r X_i  \geq 2t+2rp\Big)  &\leq \Pr\Big(\sum_{i=1}^r Y_i+Z_i  \geq 2t+2rp\Big) \nonumber \\ 
		&\leq \Pr\Big(\sum_{i=1}^r Y_i  \geq t+rp\Big) + \Pr\Big(\sum_{i=1}^r Z_i  \geq t+rp\Big) \notag\\
		& \leq 2\exp \left ( \frac{-t^2}{2r}\right ) \;,\nonumber 
		\end{align}
		where the last inequality follows from Claim~\ref{claim:yz}. 
		Replacing $t$ with $t/2$ and using Claim~\ref{claim:x-tail} proves~\eqref{eq:good-rgeneral-a}.
\end{proof}

\subsection{Derandomization}

Lemma~\ref{lem:good-rgeneral} states that if a pattern is chosen according to the distribution $\mathcal{D}^{(r)}(M,L)$, then it is \spec-causal with probability that is close to $1$, regardless of the choice of \spec. The following lemma shows that it is possible to partially derandomize the distribution, at little loss in the success probability.

\begin{lemma}\label{lem:good-rgeneral-derandomized}
Let $N\geq 1$ be an integer and $\eta>0$. Let $L,B,R_{in},r$ be integer such that 
\begin{equation}\label{eq:parameters}
B\,=\,O(N^{\eta})\;,\quad  L\,=\,O(N^{4/7+\eta})\;,\quad\text{and}\quad r\,=\,O(N^{2/7-2\eta})\;,
\end{equation}
Then using $O(\log^2 N)$ uniformly random bits it is possible to sample from a distribution $\widetilde{D}^{(r)}$ on input patterns $\pattern$ for $\grid_N$ such that for any circuit specification $\spec$ that is $(B,R_{in},0)$-bounded, $\pattern$ is $\spec$-causal with probability $1-O(N^{-\eta})$, and a random $(u,\Gamma)\in \pattern$ is valid with probability $O(R_{in}/N^2 + N^{-\eta})$. 
\end{lemma}

\begin{proof}
The choice of parameters made in the lemma is such that 
 $r^3 B L^2/N^{2} = O(N^{-\eta})$ and $r^2B/L = O(N^{-\eta})$, so Lemma~\ref{lem:good-rgeneral} gives that $\pattern$ sampled according to $D^{(r)}$ is $\spec$-causal with probability $1-O(N^{-\eta})$, and a random $(u,\Gamma)\in \pattern$ is valid with probability $O(R_{in}/N^2)$. 

A pattern in the support of $\mathcal{D}^{(r)}(M,L)$ can be specified using $O(r\log(N))$ uniformly random bits: for each $i\in\{1,\ldots,r\}$, there are $O(\log N)$ random bits to specify the  locations of the $u^{(i)}_j$, and $O(\log N)$ additional bits to specify the star that connects the $u^{(i)}_j$. Given any choice of such random bits, and a fixed circuit specification, by Savitch's theorem it is possible to decide whether the pattern is \spec-causal in $O(\log^2 N)$ space, given read-only access to the circuit graph determining \spec. This allows us to apply the INW pseudorandom generator for small-space circuits~\cite{impagliazzo1994pseudorandomness} with $O(\log^2 N)$ seed to obtain the claimed result, with the additional error $O(N^{-\eta})$ being due to the pseudorandom generator. 
\end{proof}

%-------------------------%
\section{Circuit games}
\label{sec:circuit-games}

%!TEX root = ./main.tex

Let $N\geq 1$ be an integer grid size. Let $r\geq 1$ be an integer number of repetitions. Let $\game$ be an $(\ell,k,m)$ stabilizer game. In this section we design a circuit game $G=G_{\game,N,r}$ associated with $(\game,N,r)$ in a way that the circuit game has similar completeness and soundness properties as $\game$ (more precisely, as a rotated, stretched game obtained from $\game$, using the stars in  an input pattern $\pattern$ associated with $(\game,N,r)$ that is provided as input to the circuit to define the length of the stretches; see Section~\ref{sec:rotated-game} for the definition of rotated and stretched games). 

%Throughout, we consider a fixed Boolean circuit $\m C$ which has inputs in $2^{\m I} \times \{0,1\}^{\grid_N}$ and outputs in $\{0,1\}^{\grid_N}$.

We first give a general definition that specifies what we mean by a ``circuit game''. 

\begin{definition}[Circuit game]
Given input and output sets ${\m I}$, ${\m O}$ respectively, a \emph{circuit game} is a relation $\mathcal{R} \subseteq \mathcal{I} \times \mathcal{O}$, together with a probability distribution $\pi$ on $\mathcal{I}$.
We say that a circuit $\m C$ wins the circuit game $(\mathcal{R},\pi)$ with probability $p$ if, on average over an input $x\in\m I$ sampled according to $\pi$, the circuit returns an output $y\in\m O$ such that $(x,y)\in\m R$ with probability $p$.
\end{definition}
	
To specify the relation associated with the circuit game we will construct from $\game$ it is convenient to first introduce a quantum circuit that succeeds in the game with certainty. This is done in Section~\ref{sec:game-completeness}. In Section~\ref{sec:game-def} we give the definition of the circuit game. 

\subsection{Definition and completeness}
\label{sec:game-completeness}

Informally, the game $G_{\game,N,r}$ is obtained by ``planting $r$ copies of $\game$ in the grid $\grid_N$''. 
Let $\pattern$ be an input pattern associated with $(\game,N,r)$ (see Definition~\ref{def:pattern}). Let $k = \max_j k_j$ be the maximum number of qudits used by a player in the honest strategy for $\game$, $\ket{\psi}$  the $(k\ell)$-qubit state used in the strategy (padded if needed), and $D$ the depth of a circuit that prepares $\ket{\psi}$ from $\ket{0}$. Assume that $D\geq 2$. We describe a depth $(D+1)$ quantum circuit $\mathcal{C}_{ideal}$ that takes an input from
\begin{equation}\label{eq:input-set}
\mathcal{I}\,=\, \big\{\pattern:\,\text{input pattern for $\game$}\big\}\times \big\{\{(x^{(i)}_1,\ldots,x^{(i)}_\ell)\}_{1\leq i \leq r}:\,\text{$r$-tuple of queries in $\game$}\big\}\;,
\end{equation}
 and returns a string in the output set
\begin{equation}\label{eq:output-set}
\mathcal{O} \,=\, \prod_{i=1}^r\,\prod_{j=1}^\ell \Big((\Z_d^{2 k})^{\Gamma_j^{(i)}} \times (\Z_d^{m_j})^{u_j^{(i)}}\Big)\;.
\end{equation}
In~\eqref{eq:output-set} we have used the vertices in $\Gamma_j^{(i)}$ and the $u_j^{(i)}$ to label indices of the elements of $\mathcal{O}$. Note that these are always distinct. 
 In Section~\ref{sec:geometric-local} we show how to modify the format for the input and the output in a way that the circuit can be made geometrically local on a 2D grid. 

The computation performed by the circuit $\mathcal{C}_{ideal}$ proceeds in three stages: 
\begin{itemize}
\item In the first stage, the circuit initializes a lattice of qudits as follows. Each vertex in $\grid_N$ is associated with $4k$ qudits, organized in $4$ groups of $k$ that we call the ``left'', ``right'', ``top'' and ``bottom'' groups associated with that vertex. Each of these groups is initialized in a maximally entangled state with the group from the neighboring grid vertex that is closest to it, i.e. the ``top'' group at vertex $(i,j)$ is associated with the ``bottom'' group at vertex $(i,j+1)$, etc. In addition, for each center location $g^{(i)}$ of a star $\Gamma^{(i)}$ in $\pattern$, the circuit creates the state $\ket{\psi}$ on the $(k\ell)$ qudits associated with the $\ell$ vertices $(g^{(i)}_1,\ldots,g^{(i)}_\ell)$; for each vertex $g^{(i)}_j$, a group  of qudits is used that is not connected to the next vertex in the path $\Gamma_j^{(i)}$. (This replaces the creation of the maximally entangled state, for that group of qudits.) This step can be implemented in depth $\max(2,D)$.
\item In the second stage, the circuit implements an entanglement transfer protocol as described in Section~\ref{sec:entanglement-transfer}, using each of the $\ell$ simple paths that form a star $\Gamma^{(i)}$ from $\pattern$ to route the qudits of $\ket{\psi}$. The measurement outcomes in $(\Z_d)^{2k}$ from the teleportation measurements obtained at each vertex in $\Gamma^{(i)}$ are recorded at the location at which they are obtained, and will eventually form part of the output of the circuit. This step can be completed in depth $1$.
\item In the last stage, the circuit implements the honest quantum strategy for the game $\game$, using locations $u^{(i)}_j$ indicated in $\pattern$ to specify the $k$ qudits to be used by the $j$-th player (the group used is the one closest to the endpoint of the path $\Gamma_j^{(i)}$), and $x^{(i)}_j$ as the player's question. The outcomes obtained are returned as part of the output. This step can be completed in depth $1$, and can be executed in parallel with the previous step. 
\end{itemize}

The following lemma states that outputs generated by this circuit satisfy the win condition for an associated rotated, stretched game. 

\begin{lemma}\label{lem:completeness}
 Let $\game$ be an $(\ell,k,m)$ stabilizer game, $1\leq r \leq N$, and $\pattern$ an input pattern associated with $(\game,N,r)$. Let $x^{(1)},\ldots,x^{(r)}$ be an arbitrary tuple of $r$ queries for $\game$. For all $i\in\{1,\ldots,r\}$ and $j\in\{1,\ldots,\ell\}$ let 
$(r^{(i)}_j,a_j^{(i)}) \in (\Z_d^{2k})^{\Gamma_j^{(i)}} \times \Z_d^{m_j} $
be the outputs generated by an execution of $\mathcal{C}_{ideal}$ on input $\pattern$ and $(x^{(1)},\ldots,x^{(r)})$. Then for any $i\in\{1,\ldots,r\}$, $\{(r^{(i)}_j,a_j^{(i)})\}_{j\in\{1,\ldots,\ell\}}$ is a valid $\ell$-tuple of answers for the players in the   rotated stretched game $\game_{\Gamma^{(i)}}^{S,R}$, on query $x^{(i)}$.
\end{lemma}

\begin{proof}
The lemma follows from the definition of $\mathcal{C}_{ideal}$, the properties of the entanglement transfer protocol stated in Lemma~\ref{lem:teleport}, and the definition of the rotated, stretched game. 
\end{proof}

\subsubsection{Entanglement transfer}
\label{sec:entanglement-transfer}

We introduce a simple procedure for routing entanglement along a path, such that nearest neighbors on the path have been initialized in a maximally entangled state. This is a standard calculation; for completeness we include the details. 

\begin{lemma}[Entanglement transfer I]\label{lemma:entanglement-transfer-1}
	Let $d\geq 2$ and $\ket{\epr_d} = \frac{1}{\sqrt d} \sum_{i} \ket{ii}$ a maximally entangled state on $d$-dimensional qudits. Let $a,b\in\Z_d$ and 
	\[\ket{\psi}_\reg{ABCD} \,=\,\big(X_\reg{A}^a\otimes Z_\reg{B}^b \ket{\epr_d}_\reg{AB}\big) \otimes \ket{\epr_d}_\reg{CD}\]
	a maximally entangled state on four qudits. Then upon measuring the qudits in registers $\reg{B}$ and $\reg{C}$ in the Bell basis $\set{X^x\otimes Z^y \ket{\epr_d}: x,y \leq d}$, the post-measurement state is equal (up to global phase) to 
	\begin{equation}
		\big(X_\reg{A}^{a+x}\otimes Z_\reg{D}^{b-y} \ket{\epr_d}\big)_\reg{AD} \otimes \big(X_\reg{B}^{x}\otimes Z_\reg{C}^y\EPRd\big)_\reg{BC}\;.
	\end{equation}
\end{lemma}

\begin{proof}
	We evaluate the post-measurement state by computing the result of applying the measurement projector onto the state $X^x\otimes Z^y \EPRd$. Let $\w = e^{2\pi i/d}$.
	\begin{align*}
		&\left(X^{x}\otimes Z^y \proj{\epr_d}_{BC}X^{-x}\otimes Z^{-y}\right)
		X_A^a\otimes Z_B^b \ket{\epr_d}_{AB} \otimes \ket{\epr_d}_{CD}
		\\
		&=X_B^{x}\otimes Z_C^y \left( \frac{1}{d}\sum_{k,l}\ketbra{kk}{ll}\right)X_B^{-x}\otimes Z_C^{-y}
		X_A^a\otimes Z_B^b \frac{1}{d}\sum_{i,j} \ket i \ket i \ket j \ket j
		\\
		&=X_B^{x}\otimes Z_C^y \left( \frac{1}{d}\sum_{k,l}\ketbra{kk}{ll}\right)X_B^{-x}\otimes Z_C^{-y}
		\frac{1}{d}\sum_{i,j} \w^{bi}\ket {i+a} \ket i \ket j \ket j
		\\
		&=X_B^{x}\otimes Z_C^y \left(\frac{1}{d} \sum_{k,l}\ketbra{kk}{ll}\right)
		\frac{1}{d}\sum_{i,j} \w^{bi-yj}\ket {i+a} \ket {i-x} \ket j \ket j
		\\
		&=X_B^{x}\otimes Z_C^y
		\frac{1}{d^2}\sum_{i,j,k} \w^{bi-yj}\delta_{i-x,j}\ket {i+a} \ket {k} \ket k \ket j
		\\
		&=X_B^{x}\otimes Z_C^y
		\frac{1}{d^2}\sum_{j,k} \w^{b(j+x)-yj}\ket {j+x+a} \ket {k} \ket k \ket j
		\\
		&=\w^{bx} X_A^{a+x}\otimes Z_D^{b-y}\EPRd_{AD} \otimes X_B^{x}\otimes Z_C^y \EPRd_{BC} \;.
	\end{align*}
\end{proof}

\begin{lemma}[Entanglement transfer II]\label{lemma:entanglement-transfer-2}
	Let $n\geq 1$, and $\reg{L}_1,\ldots, \reg{L}_n, \reg{R}_1, \ldots, \reg{R}_n$ qudit registers such that each $\reg{L}_i$ is maximally entangled with $\reg{R}_i$. Suppose one performs $(n-1)$ Bell basis measurements on qudit pairs $(R_1,L_2),\ldots,(R_{n-1},L_n)$, so that the post measurement state of the $i$-th pair is $X^{x_i}\otimes Z^{y_i}\EPRd$. Let $x = \sum_i x_i$ and $z = \sum_i z_i$. Then the post measurement state of the remaining pair $(L_1,R_n)$ is $X^{x}\otimes Z^{-z}\EPRd$. 
\end{lemma}

\begin{proof}
	The $(n-1)$ Bell basis measurements commute, so we can think of them as being performed in sequence, performing the measurement on $(R_k, L_{k+1})$ at the $k$-th step. Using Lemma \ref{lemma:entanglement-transfer-1} and induction, one can check that after the $k$-th measurement, the qudits $(L_1, R_{k+1})$ are in post-measurement state $X^{\sum_{i=1}^k x_i}\otimes Z^{-\sum_{i=1}^kz_i}\EPRd$.
\end{proof}

\begin{lemma}[Low-depth state teleportation]\label{lem:teleport}
Let $k\geq 1$ be an integer and $\ket\psi$ a $k$-qudit state that is \emph{locally maximally mixed} in the sense that the marginal density matrix at any one qudit is the maximally mixed state. 

	Let $N\geq 1$ and $A$ and $B$ be ordered lists of vertices in $\grid_N$ with $\abs A = \abs B = k$. Let $\set{\Gamma_j}_{1\leq j \leq k}$ be a set of $k$ vertex-disjoint, even-length paths on the grid, each with one endpoint in $A$ and one endpoint in $B$. 

For any $i\in\{1,\ldots,k\}$ denote the points of $\Gamma_i$ as $v_0\cdots v_l$ with endpoints $v_0\in A, v_l\in B$. To each of the ``odd-even'' qudit pairs $(v_1,v_2),(v_3,v_4),\ldots,(v_{l-1},v_l)$ apply a gate taking the two qudit state to $\EPRd$. Next, measure each of the ``even-odd'' qudit pairs $(v_0,v_1),(v_2,v_3),\ldots,(v_{l-2},v_{l-1})$ in the Bell basis. Let $X^{a_j^i}\otimes Z^{b_j^i}\EPRd$ be the post-measurement state of the $j$-th pair. 

	Let $P$ be the weight $k$ Pauli operator supported on the qudits of $B$ such that the $i$-th qudit of $P$ acts as $X^{a^i}Z^{-b^i}$. Then the post-measurement state of the qudits in $B$ is $P\ket\psi$.
\end{lemma}

\begin{proof}
	We analyze the circuit one path at a time. 
	Fix $i\in\{1,\ldots,k\}$. Let $x_i\in A$ be one endpoint of $P_i$ and $y_i$ the other. Suppose that $\ket\phi$ is any locally maximally mixed state on a set $C$ of qubits with $x_i \in C$. Let $C' = C\setminus \set{x_i}$. Then there is a unitary $V: C' \to C_1'\otimes C_2'$ such that 
	\begin{equation}
		I \otimes V \ket{\phi}_{x_iC} = \EPRd_{x_iC_1'}\otimes \ket{\phi'}_{C_2'}\;.
	\end{equation}
	Now suppose we apply $V$ and then apply the entangling gates and measurements along $\Gamma_i$. By Lemma \ref{lemma:entanglement-transfer-2}, the state of qudits $y_i$ and $C_1'$ is 
	\begin{equation}\label{eq:16}
		X^{a^i}\otimes Z^{-b^i}\EPRd_{y_iC_1'}
		=X^{a^i}Z^{-b^i}\otimes I\EPRd_{y_iC_1'}\;.
	\end{equation}
	The equality in~\eqref{eq:16} can be verified by noticing that $Z\otimes Z\dagg$ stabilizes $\EPRd$.
	Applying $V\dagg$ gives
	\begin{equation}
		(I\otimes V\dagg)(X^{a^i}Z^{-b^i}\otimes I) \EPRd_{y_iC_1'} \otimes \ket{\phi'}_{C_2'}
		= (X^{a^i}Z^{-b^i}\otimes I) \ket{\phi}_{y_iC'}\;.
	\end{equation}
	Notice that the operation along the path commutes with $V$. Therefore, applying $V$, applying that operation, and then applying $V\dagg$ is equivalent to just applying that operation.

	Notice that the resulting state continues to be locally maximally mixed. Therefore, we can apply the above repeatedly until all of the path circuits have been applied. Then the final state is as desired.
\end{proof}

\subsubsection{Geometric locality}
\label{sec:geometric-local}

We explain how to modify the input and outputs sets $\mathcal{I}$ and $\mathcal{O}$ specified in~\eqref{eq:input-set} and~\eqref{eq:output-set} respectively in a way that both input and output  are Boolean strings of the same length that can be organized in a 2-dimensional pattern and such that the circuit described in Section~\ref{sec:game-completeness}
 can be implemented in the same depth $(D+1)$ using only geometrically local gates. 

Recall that the input to the circuit consists of an  input pattern $\pattern = \{(\Gamma^{(i)},u^{(i)})\}_{1\leq i \leq r}$ together with an $r$-tuple of queries $\{(x^{(i)}_1,\ldots,x^{(i)}_\ell)\}_{1\leq i \leq r}$ in $\game$. Recall also that we think of the circuit as being organized on an $N\times N$ grid of vertices, such that each vertex contains $4$ groups of $k$ qudits, each group facing one of the vertex' nearest neighbors on the grid. We index the input and output sets by grid vertices, with each vertex associated with an element taken from a constant-size alphabet $\Sigma$ that is defined in~\eqref{eq:def-sigma} below. 

Each star $\Gamma^{(i)}$ specifies $\ell$ paths $\Gamma^{(i)}_j$ from the central box to the noncentral boxes. Assign to one point in each noncentral box the label $u_j^{(i)}$. Also assign to $\ell$ points inside the central box the labels $g_j^{(i)}$. Extend the paths $\Gamma^{(i)}_j$ so that their endpoints are $u_j^{(i)}$ and $g_j^{(i)}$. 
We naturally distribute each question $x^{(i)}_j$ at the grid vertex indicated by $u^{(i)}_j$. 

 For each edge $(v,w)$ in the path, at vertex $v$ (resp. $w$) we include a symbol that indicates that a teleportation measurement is to be performed between the group of $k$ qudits nearest to vertex $w$ (resp. $v$). 
For any grid vertex $v$, the output of the circuit at vertex $v$ is either an answer in $\game$, $a^{(i)}_j \in (\Z_d)^{m_j}$, or a teleportation measurement outcome, which is an element of $(\Z_d)^{2k}$. A question $x^{(i)}_j$ is an element of $(\Z_d^k)^{m_j}$. 
This leads us to a circuit specification that considers the input and output sets 
\begin{equation}\label{eq:def-sigma}
\mathcal{I} \,=\, \mathcal{O} \,= \,\Sigma^{\grid_N}\;,\quad\text{where}\quad \Sigma = \{0,1\}^{2mk\lceil\log d \rceil}\;,
\end{equation} 
where $m=\max_j m_j$ and we fixed an arbitrary embedding of the natural input and output alphabets in $\Sigma$. Note that not all input strings are used; since we consider the parameters $d,m,k$ to be constants (depending only on the type of stabilizer game chosen), the cardinality of the alphabet $\Sigma$ is constant.

\subsection{Circuit game definition}
\label{sec:game-def}

Having specified the ideal (or, ``honest'') quantum circuit that we have in mind, we are ready to give a formal definition of the circuit game associated with $r$ copies of a stabilizer game $\game$. Recall the definition of the distribution $\mathcal{D}^{(r)}$ on input patterns given in Definition~\ref{def:sample-dist}. (For clarity, we omit the arguments $N,L$, for which we will eventually make an appropriate choice.)

\begin{definition}\label{def:circuit-game}
Let $\game$ be an $(\ell,k,m)$ stabilizer game and $1\leq r \leq N$ integer. The circuit game $G_{\game,N,r}$ is a game on the input and output sets defined in~\eqref{eq:def-sigma}.
The input distribution $\pi$ is obtained by independently sampling an input pattern $\pattern$ according to $\mathcal{D}^{(r)}$ and a tuple of $r$ independent queries  $(x^{(1)},\ldots,x^{(r)})$ for $\game$, and encoding them as an element of $\mathcal{I}$ as described in Section~\ref{sec:geometric-local}.  
The relation ${\m R}\subseteq \mathcal{I}\times\mathcal{O}$ is defined as the support of the output distribution of the circuit described in Section~\ref{sec:game-completeness}, when it is provided an input in the support of $\pi$. 
\end{definition} 
	
Skipping ahead, we note that in Section~\ref{sec:randomness-from-circuit} we consider a slight variation of the circuit game $G_{\game,N,r}$ from Definition~\ref{def:circuit-game}, where the $r$ query tuples to $\game$ are no longer independent and the win condition is relaxed to allow failure in some of the game instances. These modifications allow us to obtain a circuit game whose inputs can be sampled using few random bits (polylogarithmic in $N$), and that can be won with high probability by a circuit whose gates are subject to a limited amount of noise. 

\subsection{Soundness}
\label{sec:game-soundness}

We describe a reduction from circuit strategies (i.e. circuits with constant fan-in and bounded depth) in the circuit game introduced in Definition~\ref{def:circuit-game} to strategies for the players in the game $\game$. The next lemma refers to the notions of lightcone, input pattern, and circuit specification introduced in Section~\ref{sec:lightcones} and Section~\ref{sec:patterns}, and of rotated, stretched and repeated game introduced in Section~\ref{sec:rotated-game} and Section~\ref{sec:repeated-game}.

\begin{lemma}[Circuit locality implies local simulation]\label{lem:circuit-soundness}
Let $G = G_{\game,N,r}$ be a circuit game as in Definition~\ref{def:circuit-game}. Let $\pattern= \{(u^{(i)},\Gamma^{(i)})\}_{1\leq i \leq r}$ be an input pattern in the support of the input distribution for $G$. Let ${\m C}$ be a circuit with fan-in $K$ and depth $D$ that wins with probability $p$ in $G_{\game,N,r}$, for some $0\leq p \leq 1$, conditioned on the input pattern being $\pattern$. 

Let $\eta>0$. Let $L_f$ be the lightcone function obtained from the circuit graph, and $\spec = (L_f,\bad_{in},\emptyset)$ an associated circuit specification. Assume that $\pattern$ is \spec-causal and that a fraction at least $1-\delta$ of all input pairs $(u,\Gamma)\in\pattern$ are \spec-valid, for some $\delta\in[0,1]$. 

Then there exists an $(r\ell)$-player strategy in the $r$-repeated  rotated stretched game $\game' = ((\game_r)_{\Gamma}^{S})^R$, for some $\Gamma$ depending on ${\m C}$ that is defined in the proof, such that with probability at least $p$ the strategy succeeds in a fraction at least $1-\delta$ of the game instances.
\end{lemma}

\begin{figure}[htb!]
\centering%
\includegraphics[scale=0.6, angle = 0]{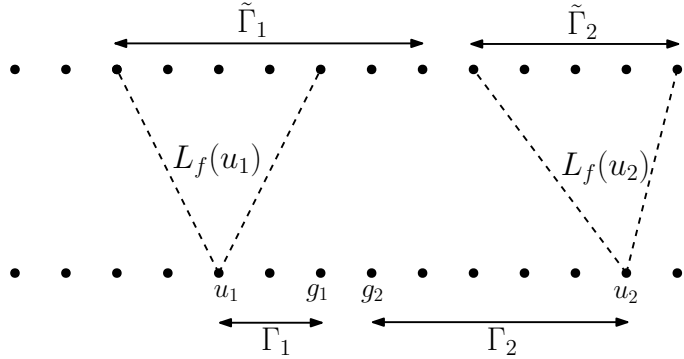}
\caption{Illustration of the construction in the proof of Lemma~\ref{lem:circuit-soundness}. A circuit with two input locations $u_1,u_2$ (case $r=1$ and $\ell=2$). The star $\Gamma$ has two simple paths linking the center vertices $g_i$ with $u_i$, for $i=1,2$. The forward lightcones of $u_1$ and $u_2$ do not intersect at any vertex in $\Gamma$. The output vertices along $\Gamma$ are partitioned into $\tilde{\Gamma}_1$ and $\tilde{\Gamma}_2$ in a way that all vertices within a lightcone of an input location $u_i$ are associated with the index $i$.}
\label{fig:lightcone}
\end{figure}

\begin{proof}
An input to the circuit $\m C$ consist of two parts: the pattern $\pattern = \{(u^{(i)},\Gamma^{(i)})\}_{1\leq i \leq r}$, and the queries $\{(x^{(i)}_1,\ldots,x^{(i)}_\ell)\}_{1\leq i \leq r}$, that are embedded in the input to the circuit as described in Section~\ref{sec:geometric-local}. By assumption there is a fraction at most $\delta$ of $(u^{(i)},\Gamma^{(i)})$ that are $\spec$-valid. For the remainder of the argument, ignore those vertices (equivalently, relabel $r$ to $(1-\delta)r$). When designing a strategy for the players in the game, the players associated to ignored vertices ignore their question and return a random answer.

The assumption that $\pattern$ is $\spec$-causal implies that for each $i\in\{1,\ldots,r\}$ each of the vertices $v\in\Gamma^{(i)}$ has a backwards lightcone that includes at most one of the input locations $u_j^{(i)}$. Moreover, all vertices in $\Gamma^{(i)} \cap \bx(u^{(i)}_j)$ have a backwards lightcone that includes no other input location than $u_j^{(i)}$. 

We define an $(r\ell)$-player strategy in the rotated, stretched game $\game'$. For each $i\in\{1,\ldots,r\}$, partition $\Gamma^{(i)}$ into sets $\widetilde{\Gamma}^{(i)}_j$, $j\in\{1,\ldots,\ell\}$, such that the only input location in the backwards lightcone of any vertex in $\widetilde{\Gamma}^{(i)}_j$  is $u_j^{(i)}$. Note that the lightcones may intersect at other grid vertices. 

Recall that by definition each wire of $\m C$ is associated with the space $\C^{d'}$, where $d' = |\Sigma|$ with $\Sigma$ the input alphabet for the circuit game. 
We now describe an unambiguous way to generate a density matrix $\sigma$ on $(\C^{d'})^{\otimes r\ell}$, together with an assignment of each qudit of $\sigma$ to a player $(i,j)$, using the circuit ${\m C}$ and the fixed input pattern $\pattern$. 

We define $\sigma$ as the output of the circuit ${\m C}$, when certain inputs have been fixed, and certain wires have been traced out. 
For all input grid vertices that are not an input location $u_j^{(i)}$, hard-wire the input to $\ket{0}$. Execute the circuit until a vertex  $v$ of the circuit graph that is in the forward lightcone of an input vertex associated with location $u_j^{(i)}$ has to be considered. Since no input has been hard-wired for that vertex, the circuit cannot proceed. There are two cases:
\begin{itemize}
\item If the forward lightcone of $v$ intersects the forward lightcone of two different input locations $u_{j}^{(i)}$, then no vertex in the forward lightcone of that location can be on any of the stars $\Gamma^{(i')}$ for $i'\neq i$ (as otherwise the vertex would be a vertex of $\Gamma$ whose backwards lightcone contains two distinct input locations). In that case, trace out the vertex. 
\item In all other cases,  vertex $v$ is in the forward lightcone of a single input location $u_j^{(i)}$. In this case, give the circuit wire associated with that vertex (in the state that it currently is) to player $(i,j)$. 
\end{itemize}
Finally, split the unassigned vertices on any path $\Gamma^{(i)}_j$ in an arbitrary way among the players; the set $\tilde{\Gamma}_j^{(i)}$ is defined as the set of vertices from the star $\Gamma^{(i)}$ assigned to the $j$-th player. All remaining unassigned vertices are traced out. 

This procedures specifies the state $\sigma$ shared by the players (see Figure~\ref{fig:lightcone} for an illustration). It remains to define their observables. Once the game starts, each player uses its input $x_j^{(i)}$ in location $u_j^{(i)}$ (encoded as an input to the circuit game, as specified in Section~\ref{sec:geometric-local}), and proceeds to complete the execution of the circuit on the qudits that it holds. If a gate has an output wire that points to a qudit that is not in the player's possession, the player measures the qudit and ignores the outcome. Finally, the player measures all qudits in the locations $\Gamma^{(i)}_j$ in the computational basis, and returns them as its answer (decoded as an answer in $\game$, as specified in Section~\ref{sec:geometric-local}). 

The fact that this strategy for the players has the same success probability in $\game'$ as the circuit $\m C$ in the circuit game $G_{\game,N,r}$ follows from the fact that the success criterion in $G_{\game,N,r}$ only involves output vertices that are along the stars $\Gamma^{(i)}$, and it can be verified that the joint operations performed by the players in the above-defined strategy correctly compute the reduced density matrix computed by the circuit on all those output vertices. Finally, using Lemma~\ref{lem:teleport} it can be verified that the success criterion in the circuit game matches the win condition for the rotated stretched  game. 
\end{proof}

\begin{remark}\label{rk:derandomized-soundness}
The proof of Lemma~\ref{lem:circuit-soundness} establishes a stronger statement than claimed in the lemma, that will be useful later. Specifically, the reduction from a circuit to a strategy for the players in $\game'$ constructed in the proof applies whenever the input pattern $\pattern$ chosen in the circuit game is \spec-causal for \spec\ the circuit specification derived from $\m C$. Moreover, whenever this is the case the reduction yields a strategy for the players that \emph{exactly} reproduces the (suitably decoded) output distribution of the circuit, on \emph{any} choice of queries $x^{(i)}$. 
\end{remark}

We end this section with the proof of Theorem~\ref{cor:exp-bgk-expsoundess}, that specifies a circuit game for which there is a very large separation between the optimal winning probabilities of classical low-depth and quantum circuits. %The presentation of Corollary \ref{cor:exp-bgk} is aided by the following easy lemma.

%\begin{lemma}\label{lemma:GHZ-k}
	%Let $\GHZk k$ be the $3k$-player game which consists of $k$ independent copies of the GHZ game, and where the players win iff every copy of the GHZ game is won. Then the maximum winning probability for a classical strategy is $\omega(\GHZk k) \leq (7/8)^k$.
%\end{lemma}
%Note that this is not a parallel repetition theorem for the GHZ game (which would apply to a game with only $3$ players, where each is given one of the inputs to each of $k$ different GHZ games. )
%\begin{proof}
%Note that the winning probability for a set of classical players in a nonlocal game is upper-bounded by the winning probability of the best deterministic strategy. For a deterministic strategy in $\GHZk k$, it's clear that the events $A_i$ given by ``the $i\th$ triple of players wins their copy of GHZ'' are independent. Then the probability that all win is at most $\Prob{A_i}^{k}$. $\Prob{A_i}$ is upper bounded by the win probability of deterministic agents in a single copy of the GHZ game, which is known to be $7/8$. 
%\end{proof}

\begin{reptheorem}{cor:exp-bgk-expsoundess}[Exponential Soundness, restated]
	Let $N\geq 1$ be an integer. There exists a circuit game $G_N$ such that the input and output sets are $\grid_{N'}$ for some $N'=O(N)$, and such that the following hold:
	\begin{itemize}
		\item (Completeness) There exists a depth-$4$ quantum circuit that succeeds at $G_N$ with probability $1$. 
		\item (Soundness) For any $\eta < 1/7$ there is a constant $c>0$ depending on $\eta$ only such that any classical circuit of depth $D\leq c\log N$ has success probability at most $\exp(-\Omega(N^{\eta}))$ in the game.
	\end{itemize}
\end{reptheorem}

\begin{proof}
	Let $G_N = G_{\game,N,r}$ be the circuit game introduced in Definition~\ref{def:circuit-game}, where the stabilizer game $\game$ is instantiated as the GHZ game from Definition~\ref{def:ghz}, and let $\pi$ and $\pattern$ be the input distribution and input pattern introduced in Definition~\ref{def:circuit-game} respectively. 
	
	\emph{Completeness:}  By Lemma \ref{lem:completeness} the circuit $\mathcal{C}_{ideal}$, as defined from the game $G_N$ at the beginning of Section \ref{sec:game-completeness}, wins the  game $G_N$ with probability at least $(\omega^*(\game))^r$, where $\omega^*(\game)$ is the optimal entangled winning probability for the stabilizer game $\game$.  Since $\omega^*(\game) = 1$ it follows that $\mathcal{C}_{ideal}$ succeeds at $G_N$ with probability 1.  
		As noted after Definition~\ref{def:ghz}, the shared state $\ket \psi_{GHZ}$ in the optimal strategy for the GHZ game can be prepared starting from $\ket{000}$ by a circuit of depth 3.  It follows by definition that $\mathcal{C}_{ideal}$ has depth at most 4.  This establishes the completeness part of the corollary. 
\newcommand{\pwin}{p_\text{win}}

	\emph{Soundness:} Fix an $\eta\in(0,\frac17)$. Let ${\m C}$ be a classical circuit with fan-in $K$ and depth $D\leq c\log N$ (for $c$ a sufficiently small constant depending on $\eta$) that wins with probability $\pwin$ in $G_{\game,N,r}$. We show that if $\pwin$ is too large, then there is a good classical strategy for the players in the game $\GHZk{\floor {r/2}}$, the $\floor{r/2}$-repeated $\GHZ$ game defined in Definition~\ref{def:n-rotated}.
	%In particular, we find $q>0$ such that
	%\begin{equation}\label{eq:condition-on-q}
		%\omega(\GHZk{\floor {r/2}})
		%\geq
		%\pwin - q.
	%\end{equation}
	%By Lemma \ref{lemma:GHZ-k}, the right-hand side of \eqref{eq:condition-on-q} is exponentially small in $r$. To prove the theorem, it will suffice to show that $q$ is also exponentially small.
	For an input pattern \pattern\ in the support of $\m D^{(r)}$, say that \pattern \emph{ has a large \spec-causal subpattern} if there exists a subpattern $\pattern' \subseteq \pattern$ such that $\pattern'$ is \spec-causal and $\abs{\pattern'} \geq \floor{r/2}$. Define $q$ as
	\begin{equation}\label{def:q-cor-bgk-exp}
	q = 1-\Prob{\pattern\text{ has a large \spec-causal subpattern}}.
	\end{equation}
%	We show that this $q$ satisfies Equation \eqref{eq:condition-on-q}.

    \newcommand\patterncore{\ensuremath{\pattern_\text{core}}}
	We first show an upper bound on $q$. For any pattern $\pattern$ let $\patterncore = \pattern_{VAL} \cap \pattern_{CAUS}$, where as in Lemma~\ref{lem:good-rgeneral2},
\begin{align*}
	\pattern_{VAL} &= \big\{(u, \Gamma) \in \pattern | (u, \Gamma) \text{ is \spec-valid} \big\}\;,
	\\
	\pattern_{CAUS} &= \big\{(u, \Gamma) \in \pattern | (u, \Gamma) \text{ is individually-\spec-causal with respect to $\pattern_{VAL}$} \big\}\;.
\end{align*}
By construction	\patterncore\ is \spec-causal.
	To prove an upper bound on $q$ it suffices to place a lower bound on the probability that $\abs{\patterncore} \geq \floor{r/2}$.  
    Recall that the classical circuit ${\m C}$ has fan-in $K$, depth $D\leq c\log N$, and circuit specification $\spec = (L_f,\bad_{in},\emptyset)$.  By Lemma~\ref{lem:lightcones}, we may choose $c$ as a function of $\eta$ such that $\spec$ is $(B, R_{in}, 0)$-bounded, where $B = O(N^{\eta})$ and $R_{in} = O(N^{2 - \eta})$.  
	By~\eqref{eq:bound-on-patterncore} we get
	\begin{equation}\label{eq:pc-1}
	\Prob{\abs{\patterncore} \geq r (1 - CrR_{in}/N^2 - 2p - 3 t/r)} \geq 1 - 4\exp \left ( -t^2/8r\right),
	\end{equation}
 	where $p = C'rB(rL^2/N^2 + \ell/L)$. Here $\ell = 3$, and the other parameters depend on $N$. 
	To set parameters, first recall that $B=O(N^{\eta})$ and $R_{in} = O(N^{2-\eta})$. Set $t = r/10$, $r=\Theta( N^\eta)$, and $L$ such that $L=O(N^{1-3\eta})$ and $L = \Omega(N^{4\eta})$, which is possible as long as $\eta<1/7$. Then $p = O(1)$ and $rR_{in}/N^2 = O(1)$. By choosing the constants appropriately, we can ensure that
	\begin{equation*}
		1 - CrR_{in}/N^2 - 2p - 3 t/r > 1/2.
	\end{equation*}
	With this choice of parameters,~\eqref{eq:pc-1} implies that $\abs{\patterncore} \geq r/2$ with probability at least
	$1 - 4\exp \left ( -C''r\right )$, for some constant $C''>0$.
To conclude, note that whenever $\pattern$ contains a \spec-causal subpattern $\pattern' \subseteq \pattern$ such that $|\pattern'|\geq r/2$ it follows from Lemma \ref{lem:circuit-soundness} and Remark~\ref{rk:derandomized-soundness} that the circuit $\mathcal{C}$ implies a strategy for $(r\ell/2)$ classical players in the repeated 
% (as in Definition~\ref{def:n-rotated})  % This definition is now cited back in the beginning of the proof when we first invoke the \GHZk k notation.
game $\GHZk {\floor{r/2}}$. Using that the maximum success probability of classical players in $\GHZ$ is $3/4$ and that the classical value multiplies under repetition (since the players are distinct) it follows that the implied strategy has success probability at most $(3/4)^{\floor{r/2}}$. %Applying a union bound with~\eqref{eq:q-bound} proves the corollary. 
\end{proof}

%-------------------------%
\section{Randomness generation}
\label{sec:randomness-generation}

%!TEX root = main.tex

In this section we give the construction of a circuit game  such that any low-depth circuit that succeeds in the game with non-negligible probability must generate outputs that have large min-entropy, even conditioned on the inputs and side information that may be correlated with the initial state of the circuit. 

The main idea for the construction is to embed a large number of copies of a simple stabilizer game $\game$ in a circuit game, as described in Section~\ref{sec:circuit-games}. (We use the Mermin $3$-player GHZ game~\cite{mermin1990extreme}, though a similar reduction could be performed starting from any stabilizer game whose quantum value is larger than its classical value.) Using the reduction from Lemma~\ref{lem:circuit-soundness}, it follows from Remark~\ref{rk:derandomized-soundness} that the output distribution of any circuit that wins with non-negligible probability in the circuit game can be deterministically mapped to a (stretched, rotated) variant of the $r$-fold parallel repetition, with $r$ independent groups of players, of $\game$. This reduction implies that, to bound the output entropy of the circuit, it suffices to place a lower bound on the output entropy of any strategy in the parallel repeated game. 

To accomplish this last step we employ the framework based on the Entropy Accumulation Theorem (EAT)~\cite{dupuis2016entropy} introduced in~\cite{arnon2018practical}, including the improvements from~\cite{dupuis2018entropy}. This framework allows to place a linear (in the number of repetitions) lower bound on the amount of min-entropy generated in the sequential repetition of a nonlocal game, using a lower bound on the function that measures the von Neumann entropy generated in a single instantiation of the game as a function of the success probability. Our setting of parallel repetition  is more constrained (thus in principle easier to analyze) than the sequential setting, but the results from~\cite{arnon2018practical,dupuis2018entropy} still give the best rates for both settings.

In Section~\ref{sec:randomness-from-stabilizer} we start by establishing a bound on the single-round randomness for the three-player GHZ game that takes the form required to apply the framework from~\cite{arnon2018practical}. In Section~\ref{sec:randomness-from-circuit} we combine this bound with the reduction from circuit games to nonlocal games shown in Section~\ref{sec:circuit-games} to deduce a family of randomness-generating circuit games.

\subsection{Randomness generation from the GHZ game}
\label{sec:randomness-from-stabilizer}

% \mnote{Should we remind the reader of why these protocols have perfect *completeness* somewhere in this section?}

We briefly recall the formalism from~\cite{arnon2018practical}, when tailored to our setting (in particular, we focus on processes specified by quantum strategies in a nonlocal game, instead of arbitrary quantum channels in~\cite{arnon2018practical}). The main definition that is needed is that of a \emph{min-tradeoff function}, which specifies a lower bound on the amount of randomness generated in a single execution of a nonlocal game, as a function of the players' probability of winning in the game. We give the definition for stabilizer games, as introduced in Definition~\ref{def:stab-game}. 

\begin{definition}[min-tradeoff function]
Let $\game$ be an $(\ell,k,m)$ stabilizer game. Fix a set of measurements $\{M_{x_j}\}$ for the $\ell$ players in the game. %Let $\Delta$ be the set of probability distributions on the output space $\Z_d^{ml}$ for the players. 
For any $\omega\in [0,1]$, let $\Sigma(\omega)$ denote the set of states $\rho_{\reg{P}_1\cdots \reg{P}_k \reg{R}}$ such that when the players' state is initialized to $\rho$ (with player $i$ being given register $\reg{P}_i$), the players' strategy wins the game with probability $\omega$.\footnote{We may without loss of generality assume that the dimension of $\reg{R}$ is no more than the sum of the dimensions of the players' private registers $\reg{P}_j$, themselves fixed by the measurements $\{M_{x_j}\}$. Therefore, the set $\Sigma(\omega)$ can be taken to be a compact set.}

Then a real affine function $f$ on $[0,1]$ is called an (affine) \emph{min-tradeoff function} for $\game$ and $\{M_{x_j}\}$ if it satisfies
\[ f(\omega) \,\leq\, \min_{\rho\in \Sigma(\omega)} H(A|QR)_{\mathcal{M}(\rho)}\;,\]
where the entropy is evaluated on the post-measurement state $\mathcal{M}(\rho)$ obtained after application of the players' measurements, and $Q$ and $A$ are random variables that represent inputs (distributed according to $\pi$) and outputs for the players in $\game$. If $f$ is a min-tradeoff function for a game $\game$ and every possible set of measurements $\{M_{x_j}\}$ for the players, then we simply say that $f$ is a min-tradeoff function for $\game$. 
\end{definition}

To illustrate the definition we apply results from~\cite{woodhead2018randomness} to give a min-tradeoff function for the Mermin GHZ game introduced in Definition~\ref{def:ghz}.
It is well-known (and easily verified) that the best classical strategy for this game succeeds with probability $\frac{3}{4}$. This in particular implies that any strategy that wins with probability strictly larger than $\frac{3}{4}$ cannot be deterministic, hence generates randomness.  The following bound shown in~\cite{woodhead2018randomness} provides a tight lower bound on the conditional entropy present in the outputs of any strategy that succeeds with sufficiently large probability.\footnote{See~(18) in~\cite{woodhead2018randomness}. The authors prove a stronger bound, that applies to the min-entropy and extends to all success probabilities in the range $[3/4,1]$. We only state the weaker bound that will be sufficient for us. To see how the bound stated in Lemma~\ref{lem:ghz-randomness} is obtained from (18) in~\cite{woodhead2018randomness}, use the replacement $\omega = \frac{1}{2}+\frac{M}{8}$.}

\begin{lemma}[\cite{woodhead2018randomness}]\label{lem:ghz-randomness}
Let $\tau = (\rho,\{M_{x_j}\})$ be a strategy with success probability $\omega \geq \frac{7}{8}$ in the game $\GHZ$, where $\rho$ is a density matrix on the provers' registers $\reg{P}_1\cdots\reg{P}_\ell$ and an arbitrary auxiliary register $\reg{R}$. Then 
\begin{equation}\label{eq:ghz-vn}
 H\big(A_1 A_2|RQ\big) \,\geq\, f_\GHZ(\omega)\,=\, -\log\Big(\frac{5}{4}-\omega + \sqrt{3}\sqrt{\big(\omega-\frac{1}{2}\big)\big(1-\omega\big)}\Big)\;,
\end{equation}
where $Q$ is a random variable that denotes the query to the players, and $A_1$, $A_2$ are random variables that denote the answers $a_1,a_2\in \Z_2$ from the first two players.
\end{lemma}

Note that for $\omega=1$, the bound~\eqref{eq:ghz-vn} gives $2$ bits of entropy, which is clearly optimal. If $\omega = 1-\eps$ for small $\eps$, the bound degrades as $2-O(\sqrt{\eps})$. 

Following the framework from~\cite{arnon2018practical}, the bound provided in Lemma~\ref{lem:ghz-randomness} already implies a linear lower bound on the entropy generated by the sequential repetition of the GHZ game. For our purposes it will be convenient to have a bound on the entropy generated by the stretched, rotated variant of the GHZ game, as introduced in Section~\ref{sec:rotated-game}. 

\begin{corollary}[min-tradeoff function for $\GHZ$]\label{cor:min-ghz}
Let $\frac{7}{8} < p_s < 1$. Let $\Gamma$ be a tuple of sets as in Definition~\ref{def:stretched-game}. Then the function $g_{p_s}:[0,1]\to\R$ defined by 
\[ g_{p_s}(q) \,=\, f_\GHZ(p_s) + (q-p_s)\frac{df_\GHZ}{d\omega}(p_s)\]
is a min-tradeoff function for the rotated stretched game $\GHZ_\Gamma^{S,R}$. 
\end{corollary}

\begin{proof}
First we observe that the bound~\eqref{eq:ghz-vn} from Lemma~\ref{lem:ghz-randomness} applies equally to the  the rotated stretched game $\GHZ_\Gamma^{S,R}$, for any fixed $\Gamma$. Indeed, fix a strategy $(\rho,\{M_{x_j}\})$ in $\GHZ_\Gamma^{S,R}$.  Using Lemma~\ref{lem:rotation-preserves-rigidty} we obtain a strategy $\tau'=(\rho,\{M'_{x_j}\})$ for $\GHZ$ that has the same success probability. Furthermore, in the coarse-graining of the strategy the answers $A_1,A_2,A_3 \in \Z_2$ in $\tau'$ are a deterministic function of the answers $A_1,A_2,A_3 \in \Z_2 \times (\Z_2^2)^{|\Gamma_1|}$ in $\GHZ_\Gamma^{S,R}$, where the second factor is for the stretched rotation string. Therefore the same bound on randomness generation that applies to $\GHZ$ applies to $\GHZ_\Gamma^{S,R}$ (as long as all outputs in the game are included, which is the case for the definition of a min-tradeoff function). 

To conclude, note that the right-hand side of~\eqref{eq:ghz-vn} is a convex function of $\omega$, hence it is at least its tangent at any point.
\end{proof}

Recall that our goal is to generate a large amount of randomness by requiring a circuit to play multiple copies of the game $\GHZ$ in parallel. Towards this, we introduce a partially derandomized variant of the repeated game, where the inputs are chosen according to a very biased distribution in order to save on the randomness required to generate them. The resulting game is a direct analogue of the protocol for randomness expansion from the CHSH game given in~\cite{dupuis2018entropy}.

\begin{definition}\label{def:drand-repeated-game}
Let  $r\geq 1$ be an  integer. Let $p,\gamma \in[0,1]$. Let $\Gamma$ be a tuple of sets as in Definition~\ref{def:stretched-game}. Let $\GHZ_{r,\Gamma,p,\gamma}^{S,R}$ denote the $r$-fold repetition, as in Definition~\ref{def:n-rotated}, of the rotated, stretched game $\GHZ$ with the following two modifications:
\begin{itemize}
\item The $r$ queries, instead of being sampled independently, are chosen according to the following distribution: first, select a subset $S\subseteq \{1,\ldots,r\}$ by including each element independently with probability $p$; second, select queries $x^{(1)},\ldots,x^{(r)}\in X$ such that $x^{(i)}$ is sampled as in $\GHZ$ when $i\in S$, and $x^{(i)} \leftarrow \ol{x}$ for $i\notin S$, where $\ol{x}$ is an arbitrary, but fixed, query in $\GHZ$;
\item It is only required that the win condition is satisfied for a fraction at least $(1-\gamma)$ of the tuples of answers $a^{(i)}$ such that $i\in S$ (there is no requirement for $i\notin S$).
\end{itemize}
\end{definition}

Corollary~\ref{cor:min-ghz} shows that any sufficiently successful strategy in the (rotated, stretched) game $\GHZ$ must generate outputs that contain a constant amount of entropy. It is therefore natural to expect that a strategy for the repeated game from Definition~\ref{def:drand-repeated-game} should generate a linear (in the number of repetitions) amount of entropy. The difficulty is to obtain a bound on the entropy generated, conditioned on having produced outputs that satisfy the win condition of the game, but without placing an implicit assumption on the intrinsic winning probability of the strategy (which would be difficult to estimate). Moreover, the fact that the strategy involves all players simultaneously measuring parts of the same entangled state introduces correlations that prevent a direct treatment using techniques appropriate for the simpler case of i.i.d.\ (identically and independently distributed) outputs. 

The Entropy Accumulation Theorem, as applied in~\cite{arnon2018practical}, is designed specifically to address these difficulties, and indeed guarantees that the repeated game introduced in Definition~\ref{def:drand-repeated-game} generates a linear amount of min-entropy. Here we use the improved results from~\cite{dupuis2018entropy},\footnote{The results in~\cite{arnon2018practical,dupuis2018entropy} apply to a much more general scenario, and in particular allow one to prove bounds on the entropy generated by an arbitrary sequential process, as long as it satisfies a certain Markov condition. Our setting, which considers parallel repetition, is easier, and trivially satisfies the required Markov condition.} that give good bounds even when the ``test probability'' $p$ from Definition~\ref{def:drand-repeated-game} can be very small. 

\begin{lemma}\label{lem:rand-stab}
Let $r\geq 1$, $p,\gamma\in[0,1]$, and $\eps>0$. Let $(\rho,\{M_{x_j}\})$, where $\rho$ is a density matrix on the player's private registers  together with an ancilla register $\reg{E}$, be a strategy for the $(3r)$ players in $\GHZ_{r,\Gamma,p,\gamma}^{S,R}$ that succeeds with probability at least $\eps$. Let $Q=(Q^{(1)},\ldots,Q^{(r)})$ (resp. $A=(A^{(1)},\ldots,A^{(r)})$) be random variables associated with the players' answers in each copy of $\GHZ$; note that each $Q^{(i)}$ (resp. $A^{(i)}$) is itself a $3$-tuple. Let $\rho_{\reg{AQE}}^s$ denote the state of $\reg{AQE}$ conditioned on the players succeeding in the game (the players' private registers are traced out). Then 
\begin{equation}\label{eq:entropy-r-round}
\Hmin^\eps(A|QE)_{\rho^s} \geq f_{\GHZ}(1-\gamma) r - O\Big(\frac{r}{\sqrt{p}} \Big(\ln\frac{1}{\eps \Tr(\rho^s) }\Big)^{1/2}\Big)\;.
\end{equation}
\end{lemma}

Note that the lower bound provided in~\eqref{eq:entropy-r-round} is non-trivial as soon as $p=\Omega(\log N/r)$ and $\eps$ is at least inverse-polynomially large (smaller $\eps$ is also possible, but requires a larger $p$).  %For small $\gamma$, $ f_{\GHZ}(1-\gamma) = 2 - O(\sqrt{\gamma})$. 

\begin{proof}
The proof is identical to the proof of~\cite[Theorem 6.1]{dupuis2018entropy}, that applies to the CHSH game. The only change needed is to use the min-tradeoff function from Corollary~\ref{cor:min-ghz} instead of the min-tradeoff function $g^*$ for the CHSH game used in~\cite{dupuis2018entropy}. The bound~\eqref{eq:entropy-r-round} follows from the bound stated in~\cite[Theorem 6.1]{dupuis2018entropy} by noting that $\frac{df}{d\omega}(1-\gamma)$ is bounded by a universal constant.  
\end{proof}

\subsection{Randomness generation from low-depth circuits}
\label{sec:randomness-from-circuit}

Using the technique to embed a stabilizer game into a circuit game described in Section~\ref{sec:circuit-games} we can leverage the randomness generation results from the previous section to obtain a family of circuit games for certified randomness expansion. First we define the circuit games that we consider by introducing a randomness-efficient, noise-tolerant modification of the game from Definition~\ref{def:circuit-game}. Even though we will eventually instantiate the definition with $\game = \GHZ$, we give the definition for a general stabilizer game. 

\begin{definition}\label{def:randomness-game}
Let $\game$ be an $(\ell,k,m)$ stabilizer game and $N\geq 1$ an integer.  Let $L,B,R_{in},r$ be parameters that satisfy the constraints~\eqref{eq:parameters} for some constant $\eta>0$. Let $p,\gamma\in [0,1]$. The circuit game $\widetilde{G}_{\game,N,r,p,\gamma}$ is defined as the game $G_{\game,N,r}$ with the following modifications:
\begin{enumerate}
\item The input pattern $\pattern$ is sampled according to the distribution $\widetilde{D}^{(r)}$ from Lemma~\ref{lem:good-rgeneral-derandomized};
\item The tuple of queries  $(x^{(1)},\ldots,x^{(r)})$ for $\game$ is sampled by first, selecting a subset $S\subseteq \{1,\ldots,r\}$ by including each element independently with probability $p$; second, selecting queries $x^{(1)},\ldots,x^{(r)}\in X$ such that $x^{(i)}$ is sampled as in $\game$ when $i\in S$, and $x^{(i)} \leftarrow \ol{x}$ for $i\notin S$, where $\ol{x}$ is an arbitrary, but fixed, query in $\game$;
\item It is only required that the win condition is satisfied for a fraction at least $(1-\gamma)$ of the tuples of answers $a^{(i)}$ such that $i\in S$ (there is no requirement for $i\notin S$).
\end{enumerate}
\end{definition}

As shown in Lemma~\ref{lem:good-rgeneral-derandomized}, it is possible to sample an input pattern from $\widetilde{D}^{(r)}$ using $O(\log^2 N)$ random bits. In addition, it is possible to sample inputs as in Definition~\ref{def:randomness-game} using $O(H(p)r)$ random bits. So, if $p = \Theta(\log N/r)$, then it is possible to sample inputs in $\widetilde{G}_{\game,N,r,p,\gamma}$ using $O(\log^2 N)$ random bits. (To this count, one may add $O(\log^3 N)$ random bits, sufficient to extract near-uniform random bits from the output of the circuit by using a seed-efficient randomness extractor~\cite{de2012trevisan}.)

The following theorem places a lower bound on the amount of randomness generated by a circuit that succeeds with non-negligible probability in the circuit game from Definition~\ref{def:randomness-game}, when the stabilizer game $\game$ is instantiated as the $\GHZ$ game from Definition~\ref{def:ghz}.

\begin{theorem}\label{thm:main-randomness}
Let $r,p,\gamma,N$ be as in Definition~\ref{def:randomness-game}, for some $\eta>0$. Let $\eta'>0$. Let $D$ be such that $D \leq c \log N$, for some sufficiently small constant $c>0$, depending on $\eta,\eta'$.
Let $\m C$ be a (classical or quantum) circuit with gates of constant fan-in and depth at most $D$. Assume that ${\m C}$ succeeds in the game $\widetilde{G}_{\GHZ,N,r,p,\gamma}$ with probability at least $\delta$, for some $\delta = \Omega(N^{-\eta'})$. Suppose the circuit is executed on its input, described by random variable $I$, as well as a an auxiliary state $\ket{\psi}_\reg{CE}$ such that the circuit acts on register $\reg{C}$, and the register $\reg{E}$ is available to an adversary. Let $O$ be a random variable that represents the circuit output, $O=\mathcal{C}(I)$. Let $\rho^s_{\reg{OIE}}$ denote the state of the inputs, outputs, and side information, conditioned on the circuit winning in the circuit game. Then there exists an $\eta'''>0$ such that for any $\eps = \Omega(N^{-\eta'''})$, it holds that
\begin{equation}\label{eq:min-lb}
\Hmin^\eps(O|IE)_{\rho^s} \geq \big(\kappa- f(\gamma)\big) r - O\Big(\frac{r}{\sqrt{p}}\Big( \ln\frac{1}{\eps \Tr(\rho^s) }\Big)^{1/2}\Big)\;,
\end{equation}
where the implicit constant depends on $\eta,\eta',\eta''$.
\end{theorem}

\begin{proof}
Let \spec\ be the circuit specification obtained from the circuit $\m C$. It follows from  
Lemma~\ref{lem:lightcones} that by  
choosing the constant $c$  small enough, 
we force \spec\ to be $(B,R_{in},0)$-bounded where 
$R_{in}=O(N^{2-\eta})$ and $B= O(N^{\eta})$. 
With this choice of parameters, it follows from Lemma~\ref{lem:good-rgeneral-derandomized} and Lemma~\ref{lem:good-rgeneral} that a pattern \pattern\ sampled from $\widetilde{D}^{(r)}$ is \spec-causal with probability at least $1-O(N^{-c'})$, for some $c'> 0$. 
Since $\m C$ succeeds in $\widetilde{G} = \widetilde{G}_{\GHZ,N,r,p,\gamma}$ with probability $\delta$, it follows that conditioned on success of $\m C$, the pattern $\pattern$ chosen as part of the input is \spec-causal with probability at least $1-O(N^{-c'}\delta^{-1})$.

% for a fraction at least $\delta/2$ of patterns $\pattern$ sampled according to $\widetilde{D}^{(r)}$, $C$ succeeds in $G$, conditioned on the input pattern being $\pattern$, with probability at least $\delta/2$. 

Let $\delta'>0$ and let $\pattern$ be a pattern which is \spec-causal on which ${\m C}$ succeeds with probability at least $\delta'$ when the pattern \pattern\ is fixed. For any such pattern, Lemma~\ref{lem:circuit-soundness} together with Remark~\ref{rk:derandomized-soundness} imply that the circuit's behavior can be simulated by a strategy for the players in the stretched rotated game $\GHZ_{r,\Gamma,p,\gamma}^{S,R}$, for some collection of sets $\Gamma$ that is determined from $\spec$. Using Lemma~\ref{lem:rand-stab} it follows that, conditioned on the input $I$ to the circuit being of the form $(\pattern,x)$ for some query $x\in  X$, the lower bound~\eqref{eq:entropy-r-round} holds. 

If $\pattern$ is such that $\m C$ succeeds with probability less than $\delta'$, then there is no bound on the min-entropy. However, the probability that this happens, conditioned on winning, is at most $\delta'/\delta$. (To see this, apply Bayes' rule directly.) Choosing $\delta' = \sqrt{\delta}$ we get an $\eta'''$, depending on $\eta'$, such that~\eqref{eq:min-lb} holds.
\end{proof}

%\section{Soundness against low-depth circuits}
%\input{graph-locality}

\bibliography{commute}

\newcommand{\etalchar}[1]{$^{#1}$}
\begin{thebibliography}{DHKLP18}
\expandafter\ifx\csname urlstyle\endcsname\relax
  \providecommand{\doi}[1]{doi:\discretionary{}{}{}#1}\else
  \providecommand{\doi}{doi:\discretionary{}{}{}\begingroup
  \urlstyle{rm}\Url}\fi

\bibitem[Aar18]{aaronson18}
S.~Aaronson.
\newblock Certified randomness from quantum supremacy, 2018.
\newblock Personal communication.

\bibitem[AC17]{aaronson2017complexity}
S.~Aaronson and L.~Chen.
\newblock Complexity-theoretic foundations of quantum supremacy experiments.
\newblock In \emph{Proceedings of the 32nd Computational Complexity
  Conference}, page~22. Schloss Dagstuhl--Leibniz-Zentrum fuer Informatik,
  2017.

\bibitem[AFDF{\etalchar{+}}18]{arnon2018practical}
R.~Arnon-Friedman, F.~Dupuis, O.~Fawzi, R.~Renner, and T.~Vidick.
\newblock Practical device-independent quantum cryptography via entropy
  accumulation.
\newblock \emph{Nature communications}, 9(1):459, 2018.

\bibitem[BCM{\etalchar{+}}18]{brakerski2018certifiable}
Z.~Brakerski, P.~Christiano, U.~Mahadev, U.~Vazirani, and T.~Vidick.
\newblock Certifiable randomness from a single quantum device.
\newblock \emph{arXiv preprint arXiv:1804.00640}, 2018.

\bibitem[BCP{\etalchar{+}}14]{brunner2014bell}
N.~Brunner, D.~Cavalcanti, S.~Pironio, V.~Scarani, and S.~Wehner.
\newblock Bell nonlocality.
\newblock \emph{Reviews of Modern Physics}, 86(2):419, 2014.

\bibitem[Bel64]{Bell:64a}
J.~S. Bell.
\newblock On the {E}instein-{P}odolsky-{R}osen paradox.
\newblock \emph{Physics}, 1:195--200, 1964.

\bibitem[BGK17]{bravyi2017quantum}
S.~Bravyi, D.~Gosset, and R.~Koenig.
\newblock Quantum advantage with shallow circuits.
\newblock \emph{arXiv preprint arXiv:1704.00690}, 2017.

\bibitem[BWKST]{BWKST18}
A.~Bene~Watts, R.~Kothari, L.~Schaeffer, and A.~Tal.
\newblock Exponential separation between shallow quantum circuits and unbounded
  fan-in shallow classical circuits.
\newblock \emph{QIP 2019 (to appear)}.

\bibitem[CHTW04]{CHTW04}
R.~Cleve, P.~H{\o}yer, B.~Toner, and J.~Watrous.
\newblock Consequences and limits of nonlocal strategies.
\newblock In \emph{Proc. 19th IEEE Conf. on Computational Complexity (CCC'04)},
  pages 236--249. IEEE Computer Society, 2004.

\bibitem[CZX{\etalchar{+}}18]{chen201864}
Z.-Y. Chen, Q.~Zhou, C.~Xue, X.~Yang, G.-C. Guo, and G.-P. Guo.
\newblock 64-qubit quantum circuit simulation.
\newblock \emph{Science Bulletin}, 2018.

\bibitem[DF18]{dupuis2018entropy}
F.~Dupuis and O.~Fawzi.
\newblock Entropy accumulation with improved second-order.
\newblock Technical report, arXiv:1805.11652, 2018.

\bibitem[DFR16]{dupuis2016entropy}
F.~Dupuis, O.~Fawzi, and R.~Renner.
\newblock Entropy accumulation.
\newblock \emph{arXiv preprint arXiv:1607.01796}, 2016.

\bibitem[DHKLP18]{dalzell2018many}
A.~M. Dalzell, A.~W. Harrow, D.~E. Koh, and R.~L. La~Placa.
\newblock How many qubits are needed for quantum computational supremacy?
\newblock \emph{arXiv preprint arXiv:1805.05224}, 2018.

\bibitem[DPVR12]{de2012trevisan}
A.~De, C.~Portmann, T.~Vidick, and R.~Renner.
\newblock Trevisan's extractor in the presence of quantum side information.
\newblock \emph{SIAM Journal on Computing}, 41(4):915--940, 2012.

\bibitem[Eke91]{Eke91}
A.~K. Ekert.
\newblock Quantum cryptography based on {B}ell's theorem.
\newblock \emph{Phys. Rev. Lett.}, 67:661--663, 1991.

\bibitem[Gal18]{gall2018average}
F.~L. Gall.
\newblock Average-case quantum advantage with shallow circuits.
\newblock \emph{arXiv preprint arXiv:1810.12792}, 2018.

\bibitem[GVW{\etalchar{+}}15]{giustina2015significant}
M.~Giustina, M.~A. Versteegh, S.~Wengerowsky, J.~Handsteiner, A.~Hochrainer,
  K.~Phelan, F.~Steinlechner, J.~Kofler, J.-{\AA}. Larsson, C.~Abell{\'a}n,
  et~al.
\newblock Significant-loophole-free test of {B}ell's theorem with entangled
  photons.
\newblock \emph{Physical review letters}, 115(25):250401, 2015.

\bibitem[HBD{\etalchar{+}}15]{hensen2015experimental}
B.~Hensen, H.~Bernien, A.~Dr{\'e}au, A.~Reiserer, N.~Kalb, M.~Blok,
  J.~Ruitenberg, R.~Vermeulen, R.~Schouten, C.~Abell{\'a}n, et~al.
\newblock Experimental loophole-free violation of a bell inequality using
  entangled electron spins separated by 1.3 km.
\newblock \emph{arXiv preprint arXiv:1508.05949}, 2015.

\bibitem[HM17]{harrow2017quantum}
A.~W. Harrow and A.~Montanaro.
\newblock Quantum computational supremacy.
\newblock \emph{Nature}, 549(7671):203, 2017.

\bibitem[HNS18]{huang2018explicit}
C.~Huang, M.~Newman, and M.~Szegedy.
\newblock Explicit lower bounds on strong quantum simulation.
\newblock \emph{arXiv preprint arXiv:1804.10368}, 2018.

\bibitem[INW94]{impagliazzo1994pseudorandomness}
R.~Impagliazzo, N.~Nisan, and A.~Wigderson.
\newblock Pseudorandomness for network algorithms.
\newblock In \emph{Proceedings of the twenty-sixth annual ACM symposium on
  Theory of computing}, pages 356--364. ACM, 1994.

\bibitem[Mer90a]{mermin1990extreme}
N.~D. Mermin.
\newblock Extreme quantum entanglement in a superposition of macroscopically
  distinct states.
\newblock \emph{Physical Review Letters}, 65(15):1838, 1990.

\bibitem[Mer90b]{mermin1990simple}
N.~D. Mermin.
\newblock Simple unified form for the major no-hidden-variables theorems.
\newblock \emph{Physical Review Letters}, 65(27):3373, 1990.

\bibitem[MFIB18]{markov2018quantum}
I.~L. Markov, A.~Fatima, S.~V. Isakov, and S.~Boixo.
\newblock Quantum supremacy is both closer and farther than it appears.
\newblock \emph{arXiv preprint arXiv:1807.10749}, 2018.

\bibitem[NC02]{nielsen2002quantum}
M.~A. Nielsen and I.~Chuang.
\newblock Quantum computation and quantum information, 2002.

\bibitem[RBH05]{raussendorf2005long}
R.~Raussendorf, S.~Bravyi, and J.~Harrington.
\newblock Long-range quantum entanglement in noisy cluster states.
\newblock \emph{Physical Review A}, 71(6):062313, 2005.

\bibitem[RUV13]{ReichardtUV13nature}
B.~Reichardt, F.~Unger, and U.~Vazirani.
\newblock A classical leash for a quantum system: Command of quantum systems
  via rigidity of {CHSH} games.
\newblock \emph{Nature}, 496(7446):456--460, 2013.

\bibitem[SBT{\etalchar{+}}18]{schafer2018fast}
V.~Sch{\"a}fer, C.~Ballance, K.~Thirumalai, L.~Stephenson, T.~Ballance,
  A.~Steane, and D.~Lucas.
\newblock Fast quantum logic gates with trapped-ion qubits.
\newblock \emph{Nature}, 555(7694):75, 2018.

\bibitem[SMSC{\etalchar{+}}15]{shalm2015strong}
L.~K. Shalm, E.~Meyer-Scott, B.~G. Christensen, P.~Bierhorst, M.~A. Wayne,
  M.~J. Stevens, T.~Gerrits, S.~Glancy, D.~R. Hamel, M.~S. Allman, et~al.
\newblock Strong loophole-free test of local realism.
\newblock \emph{Physical review letters}, 115(25):250402, 2015.

\bibitem[Tom15]{tomamichel2015quantum}
M.~Tomamichel.
\newblock \emph{Quantum Information Processing with Finite Resources:
  Mathematical Foundations}, volume~5.
\newblock Springer, 2015.

\bibitem[VV14]{VV13prl}
U.~Vazirani and T.~Vidick.
\newblock Fully device-independent quantum key distribution.
\newblock \emph{Phys. Rev. Lett.}, 113(14):140501, 2014.

\bibitem[WBA18]{woodhead2018randomness}
E.~Woodhead, B.~Bourdoncle, and A.~Ac{\'\i}n.
\newblock Randomness versus nonlocality in the {M}ermin-{B}ell experiment with
  three parties.
\newblock \emph{arXiv preprint arXiv:1804.09733}, 2018.

\bibitem[WBMS16]{wu2016device}
X.~Wu, J.-D. Bancal, M.~McKague, and V.~Scarani.
\newblock Device-independent parallel self-testing of two singlets.
\newblock \emph{Physical Review A}, 93(6):062121, 2016.

\bibitem[WEH18]{wehner2018quantum}
S.~Wehner, D.~Elkouss, and R.~Hanson.
\newblock Quantum internet: A vision for the road ahead.
\newblock \emph{Science}, 362(6412):eaam9288, 2018.

\end{thebibliography}

%\bibliography{../tex_headers/library}
\end{document}